%% file: bandit.tex
\documentclass[11pt]{article}
\usepackage{setspace}
\usepackage{url}
\usepackage{fullpage}
\usepackage{mathtools}

\usepackage[numbers]{natbib}
\usepackage{color}
\usepackage[suppress]{color-edits}
\addauthor{dk}{red}
\addauthor{rq}{blue}
\addauthor{lh}{green}

\usepackage[ruled]{algorithm2e}

\urldef{\mails}\path|{li.han, dkempe, rqiang}@usc.edu|
\title{Incentivizing Exploration with Heterogeneous Value of Money}
\author{Li Han \and David Kempe \and Ruixin Qiang}
\usepackage{definitions}
\begin{document}
\newcommand{\dif}{\mathrm{d}}
\newcommand{\bana}{\mathcal{U}}
\newcommand{\armnum}{\ensuremath{N}}
\newcommand{\reward}[2]{v_{#1,#2}}
\newcommand{\payment}[2]{c_{#1,#2}}
\newcommand{\sigSpace}{\Sigma}
\newcommand{\sigFun}{\varphi(r,\sig)}
\newcommand{\sigSch}{\varphi}
\newcommand{\sig}{s}
\newcommand{\opt}{\mbox{OPT}}
\newcommand{\theB}{B}
\newcommand{\inv}[1]{{#1}^{-1}}
\newcommand{\qStar}{\mathbf{q}^\ast}
\newcommand{\qHat}{\hat{\mathbf{q}}}
\newcommand{\sumSig}{\sum_{s \in \sigSpace}}
\newcommand{\pard}[1]{\frac{\partial}{\partial {#1}}}
\newcommand{\derone}{\frac{d}{d q}}
\newcommand{\dertwo}{\frac{d^2}{d q^2}}
\newcommand{\Umr}{conversion ratio\xspace}
\newcommand{\Umrs}{conversion ratios\xspace}
\newcommand{\MUF}{money utility function\xspace}
\newcommand{\MUFs}{money utility functions\xspace}
\newcommand{\MF}[1][]{\ifthenelse{\equal{#1}{}}{\mu}{\mu_{#1}}}
\newcommand{\Mf}[2][]{\ifthenelse{\equal{#1}{}}{\mu(#2)}{\mu_{#1}(#2)}}
\newcommand{\sigProb}{\mathbf{p}}
\newcommand{\DR}{R^{(\gamma)}}
\newcommand{\DC}{C^{(\gamma)}}
\newcommand{\lagobj}{Lagrangian objective\xspace}
\newcommand{\Zed}{Z}
\newcommand{\Tet}{total expected time-discounted\xspace}
\newcommand{\CDF}{F\xspace}
\newcommand{\pdf}{f\xspace}
\newcommand{\semiReg}{semi-regular\xspace}

\newcommand{\ALG}{\mathcal{A}}
\newcommand{\MP}[1]{q_{#1}}
\newcommand{\MPV}{\ve{q}}
\newcommand{\SP}[1]{p_{#1}}
\newcommand{\BUDGET}{b}

\providecommand{\Te}[2]{\ensuremath{\textbf{TES}_{\mathbf{#1},\mathcal{#2},\sigSch}}}
\providecommand{\TES}[2]{\ensuremath{\textbf{TES}_{\mathbf{#1},\mathcal{#2},
    \sigSch}}}
\providecommand{\TEold}[2]{\ensuremath{\textbf{TE}_{#1,\mathcal{#2}}}}
\providecommand{\TP}[1]{\ensuremath{\textbf{TP}_{#1,\mathcal{A}}}}
\newcommand{\ve}{\mathbf}
\maketitle


\begin{abstract}
Recently, Frazier et al.~proposed a natural model for crowdsourced
exploration of different a priori unknown options: a principal is
interested in the long-term welfare of a population of agents who
arrive one by one in a multi-armed bandit setting. 
However, each agent is myopic, so in order to incentivize him to
explore options with better long-term prospects, the principal must
offer the agent money. Frazier et al.~showed that a simple class of
policies called time-expanded are optimal in the worst case, and
characterized their budget-reward tradeoff.

The previous work assumed that all agents are equally and uniformly
susceptible to financial incentives. 
In reality, agents may have different utility for money.
We therefore extend the model of Frazier et al.~to allow agents that
have heterogeneous and non-linear utilities for money.
The principal is informed of the agent's tradeoff via a signal that
could be more or less informative. 

Our main result is to show that a convex program can be used to derive
a signal-dependent time-expanded policy which achieves the best
possible Lagrangian reward in the worst case. The worst-case guarantee
is matched by so-called ``Diamonds in the Rough'' instances; the proof
that the guarantees match is based on showing that two different
convex programs have the same optimal solution for these specific
instances.
These results also extend to the budgeted case as in Frazier et al.
We also show that the optimal policy is monotone with respect to information, i.e., the approximation ratio of the optimal policy improves as the signals become more informative. 
\end{abstract}




\input{introduction}

\input{relatedwork}
\input{preliminaries}
\input{algorithm}
\input{imposs}
\input{fullinfo}
\input{budget}

\input{concave-linear}
\input{garbling}
\input{conclusions}
\bibliography{bibliography/names,bibliography/conferences,bibliography/bibliography,bibliography/publications,bibliography/paper-specific}{}
\bibliographystyle{plainnat}

\appendix

\input{app-missing-proofs}
\input{rewarddep}
\end{document}

%% file: introduction.tex
\section{Introduction}
The goal of mechanism design is to align incentives when different
parties have conflicting interests.
In the VCG mechanism, the mechanism designer wants to maximize social
welfare whereas each bidder selfishly maximizes his own payoff.
In revenue maximization, the objectives are even more directly
opposed, as any increase in the bidders' surplus hurts the revenue
for the auctioneer.
In all of these cases, it is the mechanism's task to trade off
between the differing interests.

The phrase ``trade off'' is also frequently applied in the
context of online learning and the multi-armed bandit (MAB) problem,
where the ``exploration vs.~exploitation tradeoff'' is routinely referenced.
However, in the traditional view of a single principal making a
sequence of decisions to maximize long-term rewards, it is not clear
what exactly is being traded off against what.
Recent work by \citet{original} makes this tradeoff more explicit, by
juxtaposing a principal (with a far-sighted goal of maximizing
long-term rewards) with selfish and myopic agents.
Thus, the principal wants to ``explore,'' while the agents want to
``exploit.'' 
In order to partially align the incentives, the principal can
offer the agents monetary payments for pulling particular arms.

The framework of \citet{original} is motivated by many real-world
applications, all sharing the property that the principal is
interested in the long-term outcome of an exploration of different
options, but cannot carry out the exploration herself\footnote{To
  avoid ambiguity, we consistently refer to the principal as female
  and the agents as male.}.
Perhaps the most obvious fit is that of an online retailer with a
large selection of similar products (e.g., cameras on amazon.com);
in order to learn which of these products are best (and ensure that
\emph{future} buyers purchase the best product), the retailer
needs to rely on customers to buy and review the products.
Each customer prefers to purchase the best product for himself based
on the current reviews, whereas the principal may want to obtain additional
reviews for products that currently have few reviews, but may have the
potential of being high quality. 
Customers can be incentivized to purchase such products by
offering suitable discounts.

Other applications include crowd-sourced science projects, such as the
search for celestial objects or bird or fish counts: individuals may
prefer visiting areas with reliable sightings, while the principal
would like underexplored areas to receive more coverage. 
In fact, even research funding can be naturally viewed in this context:
while individual research groups may prefer to carry out research with
good short-term rewards, funding agencies can use grants as an
incentive to explore directions with long-term benefits.

\citet{original} explore this tradeoff under the standard
time-discounted Bayesian\footnote{Both \citet{original} and our
    work in fact consider a generalization in which each arm
    constitutes an independent Markov chain with Martingale rewards.}
multi-armed bandit model (described formally in
Section~\ref{sec:preliminaries}).
In each round, each arm $i$ has a known posterior reward distribution
$v_i$ conditioned on its history so far, and one arm is pulled based
on the current state of the arms.
The principal's goal is to maximize the \Tet reward
$R = \sum_{t=0}^\infty \gamma^t \Expect{v_{i_t}}$, 
where $\gamma$ is the time discount factor.
However, without incentives, each selfish agent would pull the
\emph{myopic} arm $i$ maximizing the immediate expected reward
$\Expect{v_i}$.
When the principal offers payments $c_i$ for pulling arms $i$,
in \cite{original}, the agent's utility for pulling arm $i$ is
$\Expect{v_i} + c_i$, and a myopic agent will choose the arm
maximizing this sum.

Implicit in this model is the assumption that all agents have the same
(one-to-one) tradeoff between arm rewards and payments.
In reality, different agents might have different and
  non-linear tradeoffs between these two, due to a number of causes.
The most obvious is that an agent with a large money
endowment (personal wealth or research funding) may not value
additional payments as highly as an agent with smaller endowment; 
this is generally the motivation for positing concave utility
  functions of money.
In the case of an online retailer, another obvious reason is that
different customers may intend to use the product for different
amounts of time or with different intensity, making the optimization
of quality more or less important. 
Concretely, a professional photographer may be much less willing to
compromise on quality in return for a discount than an amateur.

The main contribution of the present article is an
extension of the model and analysis of \citet{original} to incorporate
non-uniform and non-linear tradeoffs between rewards and money.
We assume that each myopic agent has a monotone and concave \emph{\MUF} 
$\MF : \R^+ \to \R^+$ mapping the agent's
  payment to the corresponding utility.
The utility of an agent with \MUF $\MF$ 
is quasi-linear: $\Expect{v_i} + \Mf{c_i}$.
The larger the values of $\MF$, the easier it is to incentivize the
agent with money, while an agent with $\MF \equiv 0$ 
cannot be incentivized at all.
When an agent arrives in round $t$, we assume that his \MUF
$\MF[t]$ is drawn i.i.d.~from some known distribution $\CDF$. 

An important question is then how much the principal knows about $\MF[t]$
at the time she chooses the payment vector $\ve{c}_t = (\payment{t}{i})_i$
to announce for the arm pulls.
In the worst case, the principal may know nothing about agent
$t$ as he arrives. In that case, the payment vector $\ve{c}_t$
can only depend on $\CDF$.
At the other extreme, the principal may learn the value of
$\MF[t]$ exactly. 
Then she is able to precisely control the myopic agent's decision by
setting $\ve{c}$ accordingly\footnote{If she wants arm $i$ to be
pulled, setting $c_i = \MF[t]^{-1}(\max_j \Expect{v_j} - \Expect{v_i})$
suffices.}.

Reality will typically lie between these two extreme cases. 
Both financial endowments and intended use can be partially inferred
from past searches and purchases in the case of an online retailer.
This partial information will give the principal a more accurate
estimate of the agent's \MUF value than what could be learned from the
prior distribution $\CDF$ alone, allowing her to better engineer the
incentives. 

We formally model the notion of partial information using the standard
economic notion of an exogenous signaling scheme \cite{spence1973job}.
A signaling scheme specifies how
signals are correlated with the ground truth, via a conditional distribution.

For each possible ground truth value (here: $\MF[t]$), the signaling
scheme prescribes a distribution over possible signals $\sig \in
\sigSpace$ that could be revealed. Then, a known signaling scheme
induces a posterior distribution when receiving a signal:
when signal $\sig$ is revealed to the principal, she can update her
posterior belief of the \MUF $\MF$ from $\CDF$ to a more ``accurate''
$F(\MF|s)$. 

The principal now has two goals, which stand in contrast with each
other: minimizing her \Tet payment 
$C = \sum_{t=0}^\infty \gamma^t \payment{t}{i_t}$, 
and maximizing her \Tet reward $R$. 
These two quantities in some sense capture the two objectives that
must be traded off: $R$ is the long-term reward to be maximized, while
$C$ captures the loss in immediate payoffs.
There are two natural ways of combining these two objectives:
The first is to maximize a Lagrangian objective $R - \lambda C$, 
for some constant $\lambda \in (0, 1)$.
The other is to maximize $R$ subject to a constraint on the 
\Tet payments.

\subsection*{Our Results}
In Section~\ref{sec:concave-linear}, we show that in a certain sense,
\emph{linear} functions $\MF$ constitute the worst case for the principal. 
Specifically, we show that for every distribution over functions $\MF$
and a corresponding signaling scheme, we can determine a distribution
over linear functions such that the proposed mechanisms in this
article perform at least as well under the original distribution as 
under the modified one.
On the other hand, we show that for certain MAB instances, the
original distribution does not allow strictly better mechanisms than
the distribution over linear functions.
The proof relies heavily on the techniques developed throughout the
article; however, it retroactively justifies the sole focus 
in the rest of the article on the case that
$\Mf[t]{x} = r_t \cdot x$.

Let $\opt_\gamma$ be the optimal \Tet reward if the principal is
allowed to pull arms herself.
We call a policy an $\alpha$-approximation policy if it achieves at
least an $\alpha$-fraction of $\opt_\gamma$ for \emph{every} MAB
instance. (Precise definitions are given in Section~\ref{sec:preliminaries}.)
The main result of this paper is to characterize the optimal
way to utilize partial information of agents to incentivize
exploration from them.

\begin{theorem} \label{theorem:main_at_intro}
  Let $\gamma$ be the time discount factor.
  Given a prior distribution $\CDF$ (satisfying some technical
    conditions) and signaling scheme $\sigSch$, 
  one can efficiently compute a policy \textbf{TES} and $p^*(\sigSch)$
  such that the Lagrangian reward of \textbf{TES} is a 
  $(1 - p^\ast(\sigSch) \gamma)$-approximation to $\opt_\gamma$.
  This bound is tight. 
\end{theorem}

\begin{theorem} \label{theorem:budget_at_intro}
  Given a prior distribution $\CDF$ (again, satisfying some
    technical conditions), signaling scheme $\sigSch$ and budget
    constraint $b$,  
  there exists a policy \textbf{TES} whose
  \Tet reward is a $\min_\lambda \{ 1 - p^\ast(\sigSch) \lambda + \lambda b \} - \epsilon$
  approximation to $\opt_\gamma$, while spending at most $b
  \opt_\gamma$ in expectation.
  This bound is tight.
\end{theorem}

In a sense, these theorems quantify the power of partial information
$\sigSch$ about a \MUF distribution $\CDF$ in a single number
$1 - p^\ast(\sigSch) \gamma$,  
via the approximation guarantee that can be achieved using this signal.
If this number is meaningful, more informative signaling
schemes should allow for better approximation ratios. 
Specifically, a \emph{garbling} \cite{marschak1968economic} 
of a signaling scheme $\sigSch$ is another signaling scheme $\sigSch'$
whose output is computed solely from the \emph{output} of $\sigSch$,
without knowledge of the true underlying state of the world. 
(A formal definition is given in Section~\ref{sec:garbling}.)
In this sense, $\sigSch'$ cannot contain more information than
$\sigSch$. Then, we prove the following theorem in Section \ref{sec:garbling}.


\begin{theorem} \label{theorem:garbling_at_intro}
Let $\sigSch$ and $\sigSch'$ be two signaling schemes such that
$\sigSch'$ is a garbling of $\sigSch$.
Then, $1 - p^\ast(\sigSch) \gamma \geq 1 - p^\ast(\sigSch') \gamma$. 
\end{theorem}

In Section \ref{sec:algo}, we prove the lower bound (algorithm) part of Theorem
\ref{theorem:main_at_intro} by using the idea of time-expansion of policies
from \cite{original}. 
The main idea there is to randomize between pulls of the myopically
optimal arm and the arm pulled by an optimal policy.
Contrary to \cite{original}, we now have to carefully coordinate the
randomization and payment policies for the different possible signals.
This is accomplished by a convex program:
the program is formulated
predominantly as a heuristic, designed to cancel out myopic reward
terms in the objective, which are otherwise difficult to analyze. 

In Section \ref{sec:diamond}, we prove that this heuristic is surprisingly optimal.
For the matching upper bound, we use a class of instances called
Diamonds-in-the-Rough \cite{original},
and show that the optimal policy using payments can only achieve
a Lagrangian reward of $(1 - p^\ast(\phi) \gamma) \opt_\gamma$. 
We characterize the optimal policy with a different convex program.
Using the Karush-Kuhn-Tucker (KKT) conditions for optimality of
solutions to a convex program \cite{luenberger1968optimization}, 
we relate these two convex programs and show that their solution is
actually the same, thus proving worst-case optimality.


We remark on the computational considerations of the proposed
policies. Given an explicit representation of the distribution $\CDF$
and the signaling scheme, the convex program produces --- without
knowledge of any MAB instance --- a vector of target randomization
probabilities $\MPV$ between myopic and non-myopic play, as well as
the optimal $p^*$.
When an actual MAB instance is specified, the algorithm draws on the
optimum policy $\opt_\eta$ for a different time-discount factor $\eta$.
When an agent arrives, the principal needs to identify the myopically
optimal arm and the arm pulled by $\opt_\eta$ under a subset of
the revealed information. Then, using the desired randomization, she
can easily compute the required payments.

Thus, the only computationally challenging part is to compute the arm
pulled by the optimum policy $\opt_\eta$. 
This can be accomplished by the well-known Gittins Index policy
\cite{GiJo74}, which computes an index for each arm $i$ based on the
posterior distribution of rewards (or state of the Markov chain), and
then chooses the arm with largest index. Thus, the policy avoids a
combinatorial explosion in the number of arms, but computing the index
of an arm can be non-trivial.

For the budgeted version, one first needs to find a suitable
Lagrangian multiplier $\lambda$ ensuring that the expected payment
respects the 
budget, and then find the optimum policy for that $\lambda$.
Finding $\lambda$ requires analyzing the specific MAB instance
and thus the optimal probabilities $\MPV$ are not independent of the MAB
instance any more.
Theorem~\ref{theorem:budget_at_intro} is proved in Section~\ref{sec:budget}.

Throughout this article and \cite{original}, it was assumed that the principal can observe the action that the agent took, and base the payment upon it. In many settings, the principal may only observe the arm \emph{reward} that the agent obtained, but not the actual arm pulled. For instance, when the agent is a scientist, a funding agency may be unable to tell whether a scientific result was achieved by deliberately following an ambitious agenda, or by stumbling upon it. In terms of the model, an inability to observe the action means that the payment scheme can be based only on the agent's observed reward. In Appendix B, we show that in this kind of scenario, it is impossible to incentivize the agent to pull optimal arms.

%% file: relatedwork.tex
\subsection*{Related Work}

The MAB problem was first proposed by \citet{Ro52} as a model for
sequential experiments design. 
Under the Bayesian model with time-discounted rewards, the problem is
solved optimally by the Gittins Index policy \cite{GiJo74};
a further discussion is given in
\cite{Wh80,KaVe87,Gi89,GittinsGlazebrookWeber2011}. 

An alternative objective, often pursued in the CS literature, is
regret-minimization, as initiated by \citet{LaRo85} within a Bayesian arm
reward setting. \citet{AuCeFrSc95,Auer2003} gave an algorithm with
regret bound for adversarial settings.

There is a rich literature that considers MAB problems when
incentive issues arise. A common model is that a principal has to
hire workers to pull arms, and both sides want to maximize their own
utility. \citet{singla:krause:truthful-crowdsourcing} gave a truthful
posted price mechanism. 
In \cite{wisdom-of-the-crowd,bandit-exploration}, the reward 
history is only known by the principal, and she can incentivize workers
by disclosing limited information about the reward history to the
worker. 
\citet{HoSlVa2014} used the MAB framework as a tool to design optimal
contracts for agents with moral hazard. Using the technique of
discretization, they achieved sublinear regret for the net
utility (reward minus payment) over the time horizon.
For a review of more work in the area, see the position paper
\cite{slivkins:wortman:position-paper}.

\citet{bergemann:valimaki:learning-pricing} and \citet{bolton:harris:experimentation}
consider incentive issues from another viewpoint. 
In their model, there are only two arms corresponding to two sellers
and one or multiple buyers. 
Both sellers can set a price for a single pull of their own arm. 
After each pull, sellers can adjust their prices, and buyers can
change their choices, which will eventually lead to an (efficient)
equilibrium. 

MAB problems with additional constraints and structure are also
studied in various settings.
\citet{guha:munagala:budgeted-learning,guha:munagala:bayesian-bandit,guha:munagala:switching-costs}
study a series of models with different constraints. 
In \cite{KlSlUp2008}, for arms defined within a metric space and satisfying
a Lipschitz condition, the authors found an optimal algorithm that
matches the best possible regret ratio. \citet{ratio-index} introduced
an index policy, similar to the Gittins index policy,  with constant
approximation for the time horizon constrained MAB
problem. \citet{bks:bandits-knapsacks} investigated MAB problems with
general multi-dimensional constraints. 

The role of information in markets was introduced formally by
\citet{stigler1962information}. Subsequently,
\citet{akerlof1970market,spence1973job} began the study of effects of
additional information, or signals, on the market. 
For exogenous signaling schemes,
\citet{hirshleifer1971private,bassan2003positive,lehrer2010signaling}
explored the positive and negative effects on the equilibrium of
different game settings.  

%% file: preliminaries.tex
\section{Preliminaries}
\label{sec:preliminaries}

\subsection{Multi-armed Bandits}
In a Bayesian multi-armed bandits (MAB) instance, we are
given $\armnum$ arms, each of which evolves independently as a known
Markov chain whenever pulled.
In each round\footnote{We use the terms ``round'' and ``time''
interchangeably.} $t = 0, 1, 2, \cdots$, an algorithm can only 
pull one of the arms; the pulled arm will generate a random reward
and then transition to a new state randomly according to the
  known Markov chain.

Formally, let $\reward{t}{i}$ be the random reward generated by arm
$i$ if it is pulled at time $t$.  
Let $S_{0,i}$ be the initial state of the Markov chain of the $i$-th
arm and $S_{t, i}$ the state of arm $i$ in round $t$. 
The distribution of $\reward{t}{i}$ is determined 
by $S_{t,i}$.
Then, an MAB instance consists of $\armnum$ independent Markov chains
and their initial states $\ve{S}_0 = (S_{0,i})_{i=1}^\armnum$.

In this article, we are only interested in cases where the
reward sequence for any single arm forms a \emph{Martingale}, i.e.,
\begin{align*}
\ExpectC{\ExpectC{\reward{t+1}{i}}{S_{t+1,i}}}{S_{t,i}}
& = \ExpectC{\reward{t}{i}}{S_{t,i}}.
\end{align*}
A \emph{policy} $\ALG$ is an algorithm that decides which arm to pull
in round $t$ based on the history of observations and the current
state of all arms. 
Formally, a policy is a (randomized) mapping 
$\ALG: (t, \mathcal{H}_t, \ve{S}_t) \mapsto i_t$, where
$\ve{S}_t = (S_{t,i})_{i=1}^N$ is the vector of arms' states,
$\mathcal{H}_t$ is the history up to time $t$,
and $i_t$ is the selected arm.

To evaluate the performance of a policy $\ALG$, we use standard time-discounting \cite{GiJo74}.
Let $\gamma \in (0,1)$ be the time discount factor that measures the
relative importance between future rewards and present rewards.
If a policy $\ALG$ receives an immediate reward of
$\reward{t}{i_t}$ in round $t$, it is discounted by a factor of
$\gamma^t$ and then added to the total reward. 
The \Tet reward can thus be defined as:
\begin{align*}
  \DR(\ALG)
& =\Expect[\ALG]{\sum_{t=0}^\infty \gamma^t \reward{t}{i_t}},
\end{align*}
where $\Expect[\ALG]{ \: \cdot \:}$ denotes the expectation
conditioned on the policy $\ALG$ being followed and the
information it obtained, as in \cite{original}.

Given a time discount factor $\gamma$,
we denote the optimal policy for that time discount (and also --- in a
slight overload of notation --- its \Tet reward) by $\opt_\gamma$.

We call the arm with the maximum \emph{immediate} expected reward
$\ExpectC{\reward{t}{i}}{\ve{S}_t}$ the \emph{myopic arm}. 
A policy is called myopic if it pulls the myopic arm in each round. 
The myopic policy only \emph{exploits} with no \emph{exploration},
so it is inferior to the optimum policy in general,
especially when the time-discount factor $\gamma$ is close to $1$.

\subsection{Selfish Agents}
We label each agent by the time $t$ when he arrives.  
The Markov chain state $\ve{S}_t$ and
$\ExpectC{\reward{t}{i}}{\ve{S}_t}$ are publicly known by both the
agents and the principal.

In round $t$, the principal can offer a payment $\payment{t}{i}$
for pulling arm $i$.
Incentivized by these extra payments $\payment{t}{i}$,
agent $t$ with \MUF $\MF[t]$ now pulls the arm maximizing 
$\ExpectC{\reward{t}{i}}{S_{t, i}} + \Mf[t]{\payment{t}{i}}$.
If the agent pulls arm $i_t$ at time $t$, then the principal's reward
from this pull is $\ExpectC{\reward{t}{i_t}}{S_{t, i_t}}$, 
and the agent's utility is 
$\ExpectC{\reward{t}{i_t}}{S_{t, i_t}} + \Mf[t]{\payment{t}{i_t}}$.
(\cite{original} studied the special case where $\mu_t(x) = x$ for all $x \ge 0$ and $t$.)

We assume a publicly known prior 
(whose distribution is denoted by $\CDF$ \footnote{When the \MUFs are always linear, $\CDF$ is also the cumulative distribution function of the slope.}) 
over the \MUFs $\MF$.
When a new agent arrives, his \MUF is drawn from $\CDF$ independently
of prior draws.

As discussed in the introduction, we show in
Section~\ref{sec:concave-linear} that the worst-case analysis can
without loss of generality focus on the case in which all
\MUF are linear, i.e., of the form $\MF[t](x) = r_t \cdot x$ for some
$r_t$. Therefore, apart from that section itself, we will exclusively
focus on the case of linear \MUFs.
We then identify the distribution $\CDF$ with a distribution over the
values $r_t$, which we call the \emph{\Umr} of agent $t$.
For the remainder of this article, all distributions and signals are
assumed to be over \Umrs instead of \MUFs.

\begin{lemma}
\label{lemma:concave-linear}
For every distribution $F$ over \MUFs, 
there exists another distribution $F'$ over \emph{linear} \MUFs (i.e., over \Umrs) 
such that
the optimal approximation ratio is the same for $F$ and $F'$.
\end{lemma}

The definition of \emph{optimal} approximation ratio can be found in Definition \ref{def:approx}.
 
\subsection{Signaling Scheme}
\label{sec:prelim-signaling}
We assume the existence of an exogenous signaling scheme, i.e., the
signaling scheme is given as input\footnote{This is in contrast to the
  goal of \emph{designing} a signaling scheme with certain properties.}.
When an agent with \Umr $r$ arrives, a signal
$\sig \in \sigSpace$ correlated to $r$ is revealed to the principal
according to the signaling scheme $\sigSch$; $\sigSpace$ is called the
\emph{signal space}, and we assume that it is countable.
When the signal space is uncountable, defining the posterior
probability density requires the use of Radon-Nikodym
derivatives, and raises computational and representational
issues. In Section~\ref{sec:full-information}, we consider what is
perhaps the most interesting special case: that the signal reveals the
precise value of $r$ to the principal.


Formally,
let $\sigFun$ be the probability that signal $\sig$ is revealed when
the agent's \Umr is $r$.
In this way, the signals are statistically correlated with the \Umr
$r$, and thus each signal reveals partial information about the true
value of $r$. 
After receiving the signal $\sig$, the principal updates her
posterior belief of the agent's \Umr according to Bayes
  Law\footnote{In Equation~\ref{eqn:bayes-update}, if the support of
  $r$ is finite, $f(r)$ can be replaced by the probability mass
  function.}: 
\begin{align} \label{eqn:bayes-update}
  f_\sig(r) & = \frac{ \sigFun f(r)}{\SP{s}},
\end{align}
where $f_\sig(r)$ is the PDF of the posterior belief and 
$\SP{s} = \int_0^\infty \sigFun f(r) \mathrm{d} r$ is the probability that
signal $s$ is observed.
For each signal $\sig \in \sigSpace$, 
let $F_s$ be the CDF of the corresponding posterior belief. 
As a special case, if the signaling scheme reveals no information,
then $F_s = \CDF$.


Throughout, we focus on the case when the posterior distributions
$F_s$ satisfies a condition called \emph{semi-regularity}
(this is the technical condition mentioned in
  Theorems~\ref{theorem:main_at_intro} and \ref{theorem:budget_at_intro}), 
which is defined as follows: 

\begin{definition}[Semi-Regularity]
  A distribution with CDF $G$ is called \emph{\semiReg} if 
  $\frac{1 - x}{\inv{G}(x)}$ is convex.
  (When $G$ is not invertible, we define 
   $\inv{G}(x) := \sup \{t \ge 0: G(t) \le x\}$.)
\end{definition}

Semi-regularity is a generalization of a well-known condition called
regularity, defined as follows.

\begin{definition}[Regularity]
A distribution with CDF $G$ is \emph{regular} 
if $\inv{G}(x) \cdot (1 - x)$ is concave.
\end{definition}

Lemma~\ref{lemma:semireg}, proved in
Appendix~\ref{sec:missing-proofs}, establishes that
regularity implies semi-regularity, and hence our result is more
general.

\begin{lemma} \label{lemma:semireg}
Let $G$ be a CDF. If $\inv{G}(x) \cdot (1 - x)$ is concave, then
$\frac{1-x}{\inv{G}(x)}$ is convex.
In particular, regularity implies semi-regularity.
\end{lemma}

\subsection{Policies with partial information}
The previous definition we gave of a policy did not take
information about the agent's type into account. In light of this
additional information, we give a refined definition.
In addition to deciding which arm to pull, a policy may decide the
payment to offer the agents based on the partial information obtained
from signals.
Formally, a policy is now a randomized mapping
$\ALG: (t, \mathcal{H}_t, \ve{S}_t, s_t) \mapsto \ve{c}_t$,
where $s_t$ is the signal revealed in round $t$, and
$c_{t, i}$ is the extra payment offered for pulling arm $i$ in round $t$.
After $\ve{c}_t$ is announced, a myopic agent with \Umr $r$ will
pull the arm $i_t$ that maximizes his own utility, causing that arm to
transition according to the underlying Markov chain.

The expected payment of the principal is also
time-discounted by the same\footnote{A natural justification
    for having the same discount factor is that after each round, with
    probability $1-\gamma$, the game ends.} factor $\gamma$. 
When $\ALG$ is implemented, the total expected payment will be
\begin{align*}
  \DC(\ALG) &=\Expect[\ALG]{\sum_{t=0}^\infty \gamma^t\payment{t}{i_t}}.
\end{align*}

The principal faces two conflicting objectives: 
(a) maximizing the \Tet reward $\DR(\ALG)$; 
(b) minimizing the \Tet payment $\DC(\ALG)$;
There are two natural ways of combining the two objectives: 
via a Lagrangian multiplier, or by optimizing one subject to a
constraint on the other.

In the \lagobj, the principal wishes to maximize 
$\DR(\ALG) - \lambda \DC(\ALG)$ for some constant $\lambda \in (0,1)$.
Here, $\lambda$ can also be regarded as the \Umr for the principal herself.
Alternatively, the principal may be constrained by a budget $\BUDGET$,
and want to maximize $\DR(\ALG)$ subject to the constraint that
$\DC(\ALG) \le \BUDGET$.

\subsection{Approximation Framework}
\citet{original} performed a worst-case analysis over MAB instances
and studied the (worst-case) approximation ratio.
In this article, we similarly perform a worst-case analysis with respect
to the MAB instances, while keeping an exogenous signaling scheme
$\sigSch$ and the prior $\CDF$ fixed.

\begin{definition}
\label{def:approx}
For the \lagobj of the problem, a policy $\ALG$
has \emph{approximation ratio} $\alpha$ 
under the signaling scheme $\sigSch$ and prior $\CDF$
if for all MAB instances\footnote{Note that all $\DR(\ALG)$,
$\DC(\ALG)$ and $\opt_\gamma$ depend on the MAB instance.},
\begin{align}
\DR(\ALG) - \lambda \DC(\ALG) & \ge \alpha \cdot \opt_\gamma.
\end{align}

We say that $\alpha$ is the \emph{optimal} approximation ratio if there exists a policy with approximation ratio $\alpha$ and no policies have a better approximation ratio.

Likewise, for the budgeted version, a policy has approximation ratio
$\alpha$ respecting budget $\BUDGET$ if
\begin{align}
\DR(\ALG) & \ge \alpha \cdot \opt_\gamma 
&  \DC(\ALG) & \le \BUDGET \cdot \opt_\gamma.
\end{align}
\end{definition}


%% file: algorithm.tex
\section{Lower bound: Time-Expanded Algorithm}
\label{sec:algo}
In this section, we focus on the Lagrangian objective, and
analyze \emph{time-expanded algorithms}, in a generalization of the
originally proposed notion of \cite{original}.
In a time-expanded algorithm, the principal
randomizes between offering the agents no reward (having them play
myopically), and offering the reward necessary to incentivize the
agent to play the arm $i_t^\ast$ according to a particular algorithm $\ALG$.
In the presence of signals, the randomization probabilities for the
different signals need to be chosen and optimized carefully, which is
the main algorithmic contribution in this section.
On the other hand, notice that if the posterior distribution of the \Umr
conditioned on the signal is continuous, then the randomness in the
user's type can instead be used as a randomization device, and the
principal may be able to offer incentives deterministically.


More formally, \citet{original} define a time-expanded version
$\TEold{p}{\ALG}$ of a policy $\ALG$, parameterized by a probability
$p$, as  
\CDFunction{\TEold{p}{A}(t)}{\ALG(\hat{\ve{S}}_t)}{\text{if }\Zed_t
= 1}{\argmax_i \ExpectC{\reward{t}{i}}{\ve{S}_t},}
where $\Zed_t$ is a Bernoulli$(1-p)$ variable. 
$\hat{\ve{S}}_t$ is the arm status that couples the execution of the
time-expanded policy and the policy $\ALG$, which we will formally
define later.
When $\Zed_t = 1$, with the uniform agents defined in \cite{original},
in order to incentivize an agent to pull the
non-myopic arm, the principal has to offer a payment of 
$\max_i\ExpectC{\reward{t}{i}}{\ve{S}_t}-\ExpectC{\reward{t}{i^\ast_t}}{\ve{S}_t}$,
where $i_t^\ast= \ALG(\hat{\ve{S}}_{t})$.

A time-expanded version of policy $\ALG$ with signaling scheme $\sigSch$
works as follows: at time $t$, conditioned on the received signal $s$,
the principal probabilistically offers a payment of 
$\payment{t}{i_t^\ast}$
if the agent $t$ pulls the arm $i_t^\ast$.
Notice that only two options might maximize the agent's utility:
pulling the myopic arm, or pulling the arm $i_t^\ast$ and getting the
payment. 
There is a direct correspondence between the payment
$\payment{t}{i_t^\ast}$ and the probability $\MP{s}$ that the 
agent chooses to pull the myopic arm.
We will describe this correspondence below.

First, though, we discuss \emph{which} arm $i_t^\ast$ the principal is
trying to incentivize the agent to pull.
As in \cite{original}, it is necessary for the analysis that the
execution of $\ALG$ and of its time-expanded version can be coupled.
To achieve this, in order to evaluate which arm should be pulled
next by $\ALG$, the principal must only take the information
obtained from the \emph{non-myopic pulls} into consideration.
Formally, we define $\hat{\ve{S}}_t$ as follows: 

Define the random variable
\CDFunction{\Zed_{t}}{0}{\text{agent $t$ pulls the myopic
    arm}}{1} 
and $X_{t,i} = 1$ if arm $i$ is pulled at time $t$ and $0$ otherwise.
Notice that
$\Zed_t$ is a Bernoulli variable, 
and $\Prob{\Zed_t=0}$ depends on the received signal $s$ and the
payment offered by the principal. 
Let $N_{t,i}=\sum_{0}^{t-1} \Zed_{t} X_{t,i}$ be the number of
non-myopic pulls of arm $i$ before time $t$.
Using this notation, we define $\hat{S}_{t,i}$ to be the state of
the Markov chain of arm $i$ after the first\footnote{As in
\cite{original}, in order to facilitate the analysis, this may include
myopic and non-myopic pulls of arm $i$. For instance, if arm $1$ was
pulled as non-myopic arm at times $1$ and $6$, and a myopic pull of arm $1$
occurred at time $3$, then we would use the state of arm $1$ after the pulls at times $1$ and $3$.} $N_{t,i}$ pulls in the
execution history of the time-expanded policy, and 
$\hat{\ve{S}}_{t} = (\hat{S}_{t,i})_i$.


Let $F_s$ be the posterior CDF of the agent's \Umr.
Let $x=\max_i \ExpectC{\reward{t}{i}}{\ve{S}_t}$ be the
expected reward of the myopic arm and
$y=\ExpectC{\reward{t}{i_t^\ast}}{\ve{S}_t}$ be the expected
reward of arm $i_t^\ast$.
If the principal offers a payment of $\payment{t}{i_t^\ast}$, 
then agents with \Umr $r < \frac{x-y}{\payment{t}{i_t^\ast}}$ will still
choose the myopic arm. 
Assuming that agents break ties in favor of the principal,
when $r \geq \frac{x-y}{\payment{t}{i_t^\ast}}$, they will
prefer to pull arm $i_t^\ast$. 

Conversely, in order to achieve a probability of $\MP{s}$ for pulling
the myopic arm, the principal can choose a payment of 
$\payment{t}{i_t^\ast}=\inf \{c|F_s(\frac{x-y}{c}) \le \MP{s}\}$. 
If $F_s$ is continuous at $\frac{x-y}{\payment{t}{i_t^\ast}}$,
the probability of myopic play (conditioned on the signal) is exactly
$\MP{s}$, and $\payment{t}{i_t^\ast}$ is the smallest payment
achieving this probability.
If there is a discontinuity at $\frac{x-y}{\payment{t}{i_t^\ast}}$,
then for every $\epsilon > 0$, the probability of myopic play
with payment $\payment{t}{i_t^\ast} + \epsilon$ is less than
$\MP{s}$. In that case, the principal offers a
payment of $\payment{t}{i_t^\ast}$ with probability
$\frac{1-\MP{s}}{1-F_s((x-y)/\payment{t}{i_t^\ast})}$ 
for pulling arm $i_t^\ast$, and no payment otherwise.
Now, the probability of a myopic pull will again be exactly
$\MP{s}$. 

To express the payment more concisely, we write
$\inv{F_s}(\MP{s}) = \sup \{r|F_s(r)\le \MP{s}\}$.
Then, the payment can be expressed as
$\payment{t}{i_t^\ast} = \frac{x-y}{\inv{F_s}(\MP{s})}$. 
In particular, when $F_s$ is continuous,
$\payment{t}{i_t^\ast}$ will be offered deterministically;
otherwise, the principal randomizes.

In summary, we have shown a one-to-one mapping between
desired probabilities $\MP{s}$ for myopic play, and 
payments (and possibly probabilities, in the case of discontinuities)
for achieving the $\MP{s}$.
We write $\MPV = (\MP{s})_{s \in \sigSpace}$ for the vector of all
probabilities.
The unconditional (prior) probability of playing myopically is
$\sum_{s \in \sigSpace} \SP{s} \MP{s}$,
and the expected payment 
$(x-y) \cdot \sum_{s \in \sigSpace} \SP{s} \frac{1 - \MP{s}}{\inv{F_s}(\MP{s})}$.

We can now summarize the argument above and give a formal
definition of the time-expanded version of policy $\ALG$ with
signaling scheme $\sigSch$, 

\begin{definition}
A policy \Te{q}{A} is a \emph{time-expanded version} of policy $\ALG$ with
signaling scheme $\sigSch$, if at time $t$, after receiving the signal $s$
about agent $t$'s \Umr,
the principal chooses (randomized) payments such that the
myopic arm is pulled with probability  $\MP{s}$, 
and the arm  $i_t^\ast=\ALG(\hat{\ve{S}}_t)$ is pulled
with probability  $1-\MP{s}$.

This can be achieved by offering the agent, with probability
$\frac{1-\MP{s}}{1-\sup\{F_s(r)|F_s(r)\le \MP{s}\}}$,
a payment of 
$\frac{\max_i\ExpectC{\reward{t}{i}}{\ve{S}_t}-\ExpectC{\reward{t}{i^\ast_t}}{\ve{S}_t}}{\inv{F_s}(\MP{s})}$
for pulling arm $i_t^\ast$.

Here, $F_s(r)$ is the CDF of the
posterior distribution of the agent's \Umr conditioned on signal $s$. 
\end{definition}


The key technical lemma gives a sufficient condition on
$\MPV$ that allows us to obtain a good approximation ratio of the
Lagrangian to the optimum solely in terms of 
$\sum_s \SP{s} \MP{s}$.

\begin{lemma}
\label{lem:TE}
Fix a signaling scheme $\sigSch$. If $\MPV$ satisfies $\sum_{s \in
    \sigSpace} \SP{s} \MP{s} \ge\lambda\sum_{s \in \sigSpace} \SP{s} \frac{1 -
\MP{s}}{F_s^{-1} (\MP{s})}$,  there exists a policy $\ALG$, 
such that the time-expanded policy $\Te{q}{A}$ satisfies
\begin{equation}
\DR_{\lambda}(\Te{q}{A}) 
\geq (1 - \gamma \cdot \sum_{s \in \sigSpace} \SP{s} \MP{s}) \cdot \opt_{\gamma}.
\end{equation}
\end{lemma}

The proof of Lemma~\ref{lem:TE} rests heavily on variations of the following lemmas
from \cite{original}. They are extremely straightforward modifications of Lemmas 4.2 and 3.2 from [9], and the proofs are deferred to Appendix~\ref{sec:missing-proofs}.

\begin{lemma}[Modification of Lemma 4.2 of \cite{original}]
\label{cite42}
Given a parameter $\lambda$ and a signaling scheme $\sigSch$. Let $\zeta_{t-1}=\sum_{t'<t}Z_{t'}$ be the total
number of non-myopic steps performed by the time-expanded algorithm $\Te{\MPV}{A}$
prior to time $t$, where $\MPV$  satisfies
$\sum_{s \in \sigSpace} \SP{s} \MP{s} \ge\lambda\sum_{s \in \sigSpace}
\SP{s} \frac{1 -\MP{s}}{F_s^{-1} (\MP{s})}$. 
Then, for any $0\le n\le t$,
\begin{align*}
\ExpectC[\Te{\MPV}{A}]{\reward{t}{i_t}-\lambda\payment{t}{i_t}}{\zeta_{t-1} = n}
& \ge \Expect[\ALG]{\reward{n}{i_n}}.
\end{align*}
\end{lemma}

\begin{lemma}[Modification of Lemma 3.2 of \cite{original}]
\label{cite32}
Given a parameter $\lambda$ and a signaling scheme $\sigSch$.  Assume  $\MPV$ satisfies
$\sum_{s \in \sigSpace} \SP{s} \MP{s} \ge\lambda\sum_{s \in \sigSpace}
\SP{s} \frac{1 -\MP{s}}{F_s^{-1} (\MP{s})}$, then for $\eta =
\frac{(1-p)\gamma}{1-p\gamma}$, where $p=\sum_{\sig\in \sigSpace}\SP{s}\MP{s}$, we have 
\begin{align*}
R_{\lambda}^{(\gamma)}(\Te{\MPV}{A})
& \ge \frac{1-\eta}{1-\gamma}\cdot R^{(\eta)}(\ALG).
\end{align*}
\end{lemma}

\begin{lemma}[Theorem 1.2 of \cite{original}] \label{cite12}  
Consider a fixed MAB instance (without selfish agents) with two
different time discount factors $\eta < \gamma$, and let
$\opt_\eta,\opt_\gamma$ be the optimum time-discounted reward
achievable under these discounts. 
Then, $\opt_\eta \ge \frac{(1-\gamma)^2}{(1-\eta)^2}\cdot \opt_\gamma$, 
and this bound is tight.
\end{lemma}

\begin{extraproof}{Lemma~\ref{lem:TE}}
Let $\ALG$ be the optimal policy for the time-discount factor
$\eta =\frac{(1-p)\gamma}{1-p\gamma}$, where $p=\sum_{\sig\in \sigSpace}\SP{\sig}\MP{\sig}$. We then have 
\begin{align*}
R_\lambda ^{(\gamma)}(\Te{q}{A}) 
&\stackrel{\text{Lemma}~\ref{cite32}}{\ge} \frac{1-\eta}{1-\gamma}\cdot\opt_\eta \stackrel{\text{Lemma}~\ref{cite12}}{\ge} \frac{1-\gamma}{1-\eta}\cdot\opt_\gamma \\
&= \frac{1-\gamma}{1-(1-p)\gamma/(1-p\gamma)} \cdot\opt_\gamma =
  (1-p\gamma) \cdot \opt_\gamma ,
\end{align*}
completing the proof.
\end{extraproof}

According to Lemma~\ref{lem:TE}, the approximation guarantee of the
time-expanded policy $\textbf{TES}$ is monotone decreasing in
$p = \sum_{s \in \sigSpace} \SP{s} \MP{s}$.
This suggests a natural heuristic for choosing the myopic
probabilities $\MPV$: minimize $p$ subject to satisfying the
conditions of the lemma.
This optimization can be carried out using the following
non-linear program.
Surprisingly, this na\"{\i}ve heuristic,
motivated predominantly by the need to cancel out terms in the
  proof of Lemma~\ref{lem:TE},
actually gives us the optimal approximation ratio. 
We will prove this in Section~\ref{sec:diamond}.

\begin{LP}[program:main]{minimize}{\sum_{s \in \sigSpace} \SP{s} \MP{s}}
\sum_{s \in \sigSpace} \SP{s} \MP{s} 
\ge \lambda \sum_{s \in \sigSpace} \SP{s} \frac{1 - \MP{s}}{F_s^{-1} (\MP{s})} \\
0 \le \MP{s} \le 1, \quad \mbox{ for all } s \in \sigSpace.
\end{LP}

First, notice that the optimization problem is feasible, because
$\MPV=\ve{1}$ is a trivial solution.  
Whenever $F_s$ is \semiReg, $\frac{1-x}{F_s^{-1}(x)}$ is convex.
Therefore, the feasibility region of the optimization problem
\eqref{program:main} is convex, and the problem can be solved 
efficiently \cite{boyd2004convex}.


\begin{theorem} \label{theorem:main_algo}
Given  a signaling scheme $\sigSch$, 
let $\MPV^\ast$ be the optimal solution of the
convex program \eqref{program:main}, and $p^\ast$ the optimal value. 
Let $\eta = \frac{(1-p^\ast)\gamma}{1-p^\ast\gamma}$. 
Then, $\textbf{TES}_{\mathbf{\MPV^\ast},\textnormal{OPT}_\eta}$ is a
$(1 - p^\ast \gamma)$-approximation policy to $\textnormal{OPT}_\gamma$.
\end{theorem}

This proves the first half of Theorem~\ref{theorem:main_at_intro} in the introduction.
Notice in Theorem~\ref{theorem:main_algo} that $\MPV$ can be
determined without knowledge of the specific MAB instance; only the
signaling scheme needs to be known.

%% file: imposs.tex
\newcommand{\Dime}{\ensuremath{\boldsymbol{\Delta}(\theB,\gamma)}\xspace}
\newcommand{\optdime}{OPT_{\lambda}^{(\gamma)}(\Dime)\xspace}

\section{Upper bound: Diamonds in the Rough}
\label{sec:diamond}

In this section, we show that the approximation ratio $1 - p^\ast
\gamma$ is actually tight when the distribution $F_s$ is \semiReg,
where $p^\ast$ is the value of the convex program \eqref{program:main}.
For simplicity, when $\qStar$ is clear from the context,  we let
$\textbf{TES}^\ast$ denote the policy 
$\textbf{TES}_{\qStar, \opt_\eta}$ 
where $\eta = \frac{(1 - p^\ast) \gamma}{1- p^\ast \gamma}$ 
(as in Theorem \ref{theorem:main_algo}). 
We will show that on a class of MAB instances called
Diamonds-in-the-rough \cite{original}, 
the \emph{optimal policy with payments} (defined below) can achieve
only a $(1- p^\ast \gamma)$-fraction of $\opt_\gamma$.
Therefore, not only is the analysis of $\textbf{TES}^\ast$'s
approximation ratio tight, but $\textbf{TES}^\ast$ also has the
\emph{optimal} approximation ratio $1 - p^\ast \gamma$.

\begin{definition}
The Diamonds-in-the-rough MAB instance \Dime is defined as follows.
Arm $1$ has constant value $1 - \gamma$.
All other (essentially infinitely many) arms have the following reward
distribution:
\begin{enumerate}
 \item With probability $1/M$, the arm's reward is a degenerate distribution of the constant $(1 - \gamma) \theB \cdot M$
 (good state);
 \item With probability $1 - 1/M$, the arm's reward is a degenerate
   distribution of the constant $0$ (bad
 state).
\end{enumerate}
\end{definition}
Note that if $\theB < 1$, then arm $1$ is the myopic arm.

Since \Dime is uniquely determined by $\theB$ and is just one single
instance, the optimal policy that maximizes the \lagobj,
i.e., $\DR(\mathcal{A}) - \lambda \DC(\mathcal{A})$, is
well-defined\footnote{This is in contrast to the case where the
  performance of a policy is evaluated on a \emph{class} of instances
  rather than single instance.}.
We call the policy that maximizes the \lagobj the \emph{optimal policy
with payments}, and denote it by $\optdime$.

We can solve for the optimal
policy with payments using another convex program, which we next derive.
Suppose that the optimal policy with payments has time-discounted \lagobj $V$.
In the first round, it only has two options: (a) let the agent play myopically
(i.e., pull the constant arm); (b) incentivize him to play a non-constant arm.

If option (a) is chosen and the agent pulled the constant arm, then
the principal learns nothing and faces the same situation in the
second round. 
So conditioned on the constant arm being pulled, the
principal will get $1 - \gamma + \gamma V$. 
If option (b) is chosen and a non-constant arm was pulled, then with
probability $1/M$, the non-constant arm will be revealed to be in the good 
state, and the principal does not need to pay any agent again,
obtaining value 
$(1 - \gamma) \theB \cdot M \sum_{i=0}^{\infty} \gamma^i = \theB \cdot M$; 
with probability $1 - 1/M$, the non-constant arm will be revealed to be in the bad
state, and the principal faces the same situation in the second round,
obtaining value $\gamma V$.
Recall that $c_s$ is the payment needed to ensure that the myopic arm
is played with probability at most $q_s$ when signal $s$ is revealed.
To summarize, if we set the probabilities for myopic play to 
$(q_s)_{s \in \sigSpace}$, 
then $V$ satisfies the following equation:

\begin{equation}
  \label{eq:diamond_opt}
    V =  (1 - \gamma + \gamma V) \sumSig \SP{s} q_s + \sumSig \SP{s} ( 1- q_s)(\frac{1}{M} \cdot \theB \cdot M + (1 - \frac{1}{M})
    \gamma V - \lambda c_s).
\end{equation}

Solving for $V$ while taking $M \to \infty$, we get 
$(1 - \gamma) V = (1 - \gamma) \sumSig \SP{s} q_s + \sumSig \SP{s} 
( 1- q_s)(\theB - \lambda c_s)$.
As the difference between the expected rewards of the myopic
arm and the non-myopic arm is 
$(1 - \gamma) - (1 - \gamma) \theB$, 
we have $c_s = \frac{(1-\gamma)(1 - \theB)}{\inv{F_s}(q_s)}$. 
The optimal policy with payments needs to choose the best myopic probabilities, which is equivalent to:

\begin{LP}[program:dp]
{maximize} { (1 - \gamma) \sumSig \SP{s} q_s + \sumSig \SP{s} ( 1- q_s)(\theB -\lambda \frac{(1-\gamma)(1 - \theB)}{\inv{F_s}(q_s)}) }
0 \le q_s \le 1, \quad \text{ for all } s \in \sigSpace.
\end{LP}

Notice that the objective function of program \eqref{program:dp} is
concave, so the program is convex.
Let $\qHat$ be the optimal solution to the program \eqref{program:dp}.
Denote by $\mathcal{A}(\qHat)$ the policy determined by $\qHat$.
Recall that $\qStar$ is the solution to the following convex program:

\begin{LP}{minimize}{\sum_{s \in \sigSpace} \SP{s} q_s}
  \tag{\ref{program:main}}
\sum_{s \in \sigSpace} \SP{s} q_s \ge \lambda \sum_{s \in \sigSpace} \SP{s} \frac{1 -
q_s}{F_s^{-1} (q_s)} \\
0 \le q_s \le 1, \forall s \in \sigSpace.
\end{LP}

Note that the $\qHat$ are probabilities for choosing the myopic arm
given by the above program and depend on a specific MAB instance,
i.e., \Dime. On the other hand, $\qStar$ is \emph{independent} of any
MAB instance and only depends on the signaling scheme $\varphi$ and
$\CDF$.
Lemma \ref{lemma:imposs} shows that for the right choice of $\theB$,
$\qHat$ and $\qStar$ actually coincide on the corresponding
Diamonds-in-the-rough instance.

The proof relies heavily on the KKT condition \cite{luenberger1968optimization},
so we assume that $\frac{1 - x}{\inv{F_s}(x)}$ is
continuously differentiable and use the KKT condition for
differentiable functions. When derivatives do not exists, we can use
the sub-differential versions of KKT. 
Since we would like the characterization to hold for countably
infinite signal spaces (and thus infinitely many variables), we have
to be somewhat careful about the specific notion of differentiability,
but note that standard notions such as Gateaux Differentiability
\cite{luenberger1968optimization} can be used here.



\begin{definition}[Slater condition]
For a convex program with no equality constraints, the Slater
condition is met if there exists a point $\ve{x}$ such that
$g_i(\ve{x}) < 0$ for all constraints $i$.
\end{definition}

It to easy to check that the program \eqref{program:main} satisfies
the Slater condition.

\begin{theorem}[KKT condition for infinite dimension]
  \label{thm:kkt}
Consider a program satisfying the Slater condition: minimize $f(\ve{x})$ subject to 
$g_i(\ve{x}) \le 0$ for $i = 1, \ldots, m$.
If $\ve{x}^\ast$ is the local minimum, then there exist multipliers $\mu_i$ for $i = 1, \ldots, m$
such that:
\begin{equation}
  \begin{aligned}
    \nabla f(\ve{x}^\ast) + \sum_{i=1}^m \mu_i \nabla g_i(\ve{x}^\ast)
    & = 0\\
    g_i(\ve{x}^\ast) & \le 0 \\
    \mu_i & \ge 0 \\
    \mu_i \cdot g_i(\ve{x}^\ast) & = 0
    \end{aligned}
  \end{equation}
\end{theorem}

\newcommand{\bmu}{\boldsymbol{\mu}}

\begin{lemma}
\label{lemma:imposs}
There exists a $\theB$ such that the myopic probabilities given by the
convex program \eqref{program:main} are equal to the myopic
probabilities given by program \eqref{program:dp}.
\end{lemma}

\begin{proof}
First, we observe that the boundary case $q_s = 0$ can be safely
ignored in both convex programs \eqref{program:main} and
\eqref{program:dp}. 
This is because $\frac{1}{\inv{F_s}(q)}$ approaches infinity as 
$q_s \to 0$. For program \eqref{program:main}, this violates the
non-trivial feasibility constraint; for program \eqref{program:dp},
this is clearly sub-optimal.

Next, we prove that every local optimum $\ve{q}^\ast$ is interior
for convex program (\ref{program:main}).
Applying the stationarity condition of the KKT Theorem to program
(\ref{program:main}), we know that for every $s \in \sigSpace$, the
optimal solution $\qStar$ satisfies: 
\begin{align}
- \SP{s} & = \mu_s + \sigma \Big(\lambda \SP{s} \frac{\partial}{\partial x} \frac{1 - x}{\inv{F_s}(x)}\vert_{x = q^\ast_s} - \SP{s} \Big).
\end{align}
Here, $\mu_s$ is the multiplier for the constraint $0 \le q_s \le 1$,
and $\sigma$ is the multiplier for the non-trivial constraint.
It is important that $\sigma$ is a constant that is independent of any
MAB instance.

Setting $q^\ast_s = 1$ for all $s \in \sigSpace$ is clearly not an
optimal solution as there will be slack in the constraints and
decreasing $q^\ast_s$ can improve the objective for program
\eqref{program:main}.
Hence, there is an $s \in \sigSpace$ such that $q^\ast_s < 1$. 
By complementarity, $\mu_s = 0$, so that
$- \SP{s} = \sigma (\lambda \SP{s} \pard{x}
\frac{1 - x}{\inv{F_s}(x)}|_{x = q^\ast_s} - \SP{s})$.
Canceling $\SP{s}$, we have 
$\pard{x} \frac{1 - x}{\inv{F_s}(x)}|_{x = q^\ast_s} 
= \frac{\sigma - 1}{\lambda \sigma}$.
But we know that $F_s$ is \semiReg, so 
$\pard{x} \frac{1-x}{\inv{F_s}(x)}|_{x = q^\ast_s} 
< \pard{x} \frac{1-x}{\inv{F_s}(x)}|_{x = 1} = 0$.
(Equality is impossible, as it would imply 
$\pard{x} \frac{1-x}{\inv{F_s}(x)} = 0$ for $x \in [q^\ast_s, 1]$ 
by definition of semi-regularity. 
But then, $\frac{1 - x}{\inv{F_s}(x)}$ would not be strictly monotone,
violating the fact that $F_s$ is non-decreasing.) 
Because $\frac{\sigma - 1}{\lambda \sigma} < 0$, we infer that
$\sigma \in (0, 1)$.

On the other hand, if $q^\ast_s = 1$ for some $s \in \sigSpace$, then
$\mu_s \ge 0$ and $-\SP{s} = \mu_s - \sigma \cdot \SP{s}$ 
(as $\pard{x} \frac{1-x}{\inv{F_s}(x)}|_{x = 1} = 0$), so $\sigma \ge 1$.
This would be a contradiction to $\sigma < 1$, 
so we infer that $q^\ast_s < 1$ for every $s \in \sigSpace$.

\smallskip

For program~\eqref{program:dp}, we first compute the partial
derivative of the objective w.r.t.~$q_s$; it is
\begin{equation}
\SP{s} \left(
1 - \gamma - \lambda (1 - \gamma) (1 - \theB) (1 - q_s) \pard{x}\frac{1}{\inv{F_s}(x)}|_{x = q_s} - \theB + \lambda \frac{(1 - \gamma)(1 - \theB)}{\inv{F_s}(x)}
\right).
\end{equation}
Rearranging and using the fact that 
$\pard{x} \frac{1 - x}{\inv{F_s}(x)} = (1 - x)\pard{x}
\frac{1}{\inv{F_s}(x)} - \frac{1}{\inv{F_s}(x)}$, we can rewrite it as
\begin{equation}
\label{eq:pd_of_dp}
\SP{s} \left(
1 - \gamma - \lambda(1-\gamma)(1 - \theB) \pard{x}
\frac{1 - x}{\inv{F_s}(x)}|_{x = q_s} - \theB
\right).
\end{equation}

Let $\theB \in (1 - \gamma, 1)$ solve the equation 
$\frac{1 - \gamma - \theB}{(1 - \gamma)(1 - \theB)} 
= \frac{\sigma - 1}{\lambda \sigma}$. 
A solution exists because $\sigma \in (0, 1)$ and 
$\frac{1 - \gamma - \theB}{(1 - \gamma)(1 - \theB)}$ varies continuously from $0$ to $-\infty$ as $\theB$ varies from $1 - \gamma$ to $1$.
Notice that the constant arm is still the myopic choice when $\theB < 1$.

Recall that $\qHat$ is a global, and thus local, maximum of
program \eqref{program:dp}.
Assume that $\hat{q}_s = 1$ for some $s$.
Then, $\pard{q_s} \frac{1-x}{\inv{F_s}(x)}|_{x = \hat{q}_s} = 0$,
so the partial derivative of the objective \eqref{eq:pd_of_dp} 
w.r.t.~$q_s$ will be $\SP{s} (1 - \gamma - \theB)$, 
which is negative because $\theB > 1 - \gamma$.
Thus, decreasing $\hat{q}_s$ would maintain feasibility and increase
the value of the program \eqref{program:dp}. 
Hence, any local maximum of the program \eqref{program:dp} must be an
interior point as well.

In summary, both $\qStar$ and $\qHat$ are interior.
By the KKT condition for $\qStar$, all the $\mu_i$'s are zero,
and thus 
$\pard{x} \frac{1 - x}{\inv{F_s}(x)}|_{x = q^\ast_s} 
= \frac{\sigma - 1}{\lambda \sigma}$.
Because $\qHat$ is a local optimum,  expression \eqref{eq:pd_of_dp}
must be equal to zero, so
$\pard{x} \frac{1 - x}{\inv{F_s}(x)}|_{x = \hat{q}_s} 
= \frac{1 - \gamma - \theB}{(1 - \gamma)(1 - \theB)}$. 

By the choice of $\theB$, we then obtain that
$ \pard{x} \frac{1 - x}{\inv{F_s}(x)}|_{x = q^\ast_s}
= \pard{x} \frac{1 - x}{\inv{F_s}(x)}|_{x = \hat{q}_s} 
= \frac{1 - \gamma - \theB}{(1 - \gamma)(1 - \theB)}$.
By semi-regularity of $F_s$, $\pard{x} \frac{1 - x}{\inv{F_s}(x)}$ is
negative and non-decreasing. 
Thus, it must be constant on the entire interval $[\hat{q}_s, q^\ast_s]$.

This means that the derivative \eqref{eq:pd_of_dp} of the
  objective of program~\eqref{program:dp} is constant over the entire
  interval, and we established above that it is 0 at $\hat{q}_s$.
  Thus, the objective function of program~\eqref{program:dp} is
  unchanged by replacing $\hat{q}_s$ with $q^*_s$.
By performing this operation for all $s$, we eventually obtain that
without loss of generality, $\qHat = \qStar$.
\end{proof}

Based on this lemma, we now prove the main theorem in this section. 
This also proves the second half of Theorem~\ref{theorem:main_at_intro} in the introduction.
\begin{theorem} \label{theorem:last_imposs}
The policy $\textbf{TES}^\ast$, parameterized by $\qStar$, has optimal
approximation ratio $1 - p^\ast \gamma$.
In particular, there exists a \emph{worst-case} MAB instance in which
the optimal policy with payments achieves exactly a Lagrangian reward
of a $(1 - p^\ast \gamma)$ fraction of the optimum.
\end{theorem}

\begin{proof}
Lemma~\ref{lemma:imposs} showed that for every signaling scheme,
there is a $\theB$ (a parameter of the non-constant arms) such that
$\qStar$ and $\qHat$ coincide on the MAB instance \Dime. 
It remains to show that these instances are in fact worst-case
instances, i.e., that the ratio between
$\DR_{\lambda}(\mathcal{A}(\qHat))$ and $\opt_\gamma$ in \Dime is
exactly $1 - p^\ast \gamma$.

First note that $\opt_\gamma = \frac{\theB}{1 - \gamma}$:
once a non-constant arm is revealed to be in a good state, it is optimal to
pull that arm forever.
We next show that 
$\DR_{\lambda}(\mathcal{A}(\qHat)) 
= \frac{\theB \cdot (1 - p^\ast \gamma)}{1 - \gamma}$.
Thereto, we use both convex programs and the fact that 
$\qHat = \qStar$. 
Substituting $\qStar$ into the objective of
program~\eqref{program:dp}, we obtain the following Lagrangian reward
for the time-expanded policy:

\begin{align*}
\DR_{\lambda}(\mathcal{A}(\qHat)) 
& = \frac{1}{1 - \gamma} \Big((1 - \gamma) \sumSig \SP{s} q^\ast_s + \sumSig \SP{s} ( 1- q^\ast_s)(\theB -\lambda \frac{(1-\gamma)(1 - \theB)}{\inv{F_s}(q^\ast_s)})\Big).
\end{align*}

The non-trivial constraint of program \eqref{program:main} has to be
tight; otherwise, the objective could be increased by lowering
individual $q^*_s$ values.
Therefore, $\sum_{s \in \sigSpace} \SP{s} q^\ast_s 
= \lambda \sum_{s \in \sigSpace} \SP{s} \frac{1 - q^\ast_s}{F_s^{-1}(q^\ast_s)}$,
which allows us to simplify the previous expression to
\begin{align*}
\DR_{\lambda}(\mathcal{A}(\qHat))
& = \frac{\theB (1 - p^\ast \gamma)}{1 - \gamma}
\; = \;  (1 - p^\ast \gamma) \cdot \opt_\gamma.
\end{align*} 
Thus, on this instance,
$\DR_{\lambda}(\textbf{TES}^\ast) \leq (1 - p^\ast \gamma) \opt_\gamma$. 
By Theorem \ref{theorem:main_algo}, this bound is tight.\end{proof}

%% file: fullinfo.tex
\newcommand{\THOLD}{\theta}
\section{Full Information Revelation}
\label{sec:full-information}

Our main positive results hold for the case of countable or finite
signal spaces, whereas uncountable signal spaces lead to technical
challenges. 
However, one important special case of uncountable signal spaces is
more easily handled, namely, when the principal learns the exact \Umr
$r$, i.e., $s=r$.
We show that in that case, $r$ itself can be used as the sole
randomization device, leading to a \emph{threshold policy}.
In this section, we assume that the distribution $\CDF$ is
continuous (an assumption that was not needed in Section \ref{sec:algo}).

\subsection{Optimal Time-Expanded Policy}

Our first goal will be to show that the optimal time-expanded policy
fixes a threshold $\THOLD$ and only incentivizes agents whose \Umr lies
above the threshold. Then, an optimization over threshold policies is
easy to carry out.
 
\begin{definition}
  The threshold policy $\TP{\THOLD}$ with threshold $\THOLD$ is 
  defined as follows:
When an agent with \Umr $r$ arrives, he is incentivized with
  suitable payment to pull $i_t^\ast = \ALG(\hat{\ve{S}}_t)$ if
  and only if $r \ge \THOLD$.
\end{definition}

\begin{lemma} \label{lem:one-round-threshold}
Consider a single arm pull, and a value $q \in [0,1]$. 
Among all policies that have this arm pull be myopic with
probability $q$, the one minimizing expected cost is a threshold
policy.
\end{lemma}

\begin{proof}
Consider any policy $\mathcal{P}$, and assume that
$q = \Prob{\text{$\mathcal{P}$ lets the agent play myopically}}$.
Now, define $\THOLD = \sup\{r| F(r) = q\}$, so that $F(\THOLD) = q$.

As before, we write
$x= \max_i \ExpectC{\reward{t}{i}}{\ve{S}_t,\Zed_{0:t-1}}$
and $ y= \ExpectC{\reward{t}{i_t^\ast}}{\ve{S}_t,\Zed_{0:t-1}}$.

For any $r$, let $P_r = 1$ iff $\mathcal{P}$ incentivizes agents
with \Umr r.
The expected payment of policy $\mathcal{P}$ is then
\begin{align*}
(x-y) \cdot \int_0^\infty \frac{P_r}{r} \dif F(r)
& = (x-y) \cdot \left(
  \int_{\THOLD}^\infty \frac{1}{r} \dif F(r)
+ \int_0^{\THOLD} \frac{P_r}{r} \dif F(r)
- \int_{\THOLD}^{\infty} \frac{1-P_r}{r} \dif F(r) \right) \\
& \geq (x-y) \cdot \left(
  \int_{\THOLD}^\infty \frac{1}{r} \dif F(r)
+ \int_0^{\THOLD} \frac{P_r}{\THOLD} \dif F(r)
- \int_{\THOLD}^{\infty} \frac{1-P_r}{\THOLD} \dif F(r) \right) \\
& = (x-y) \cdot \int_{\THOLD}^\infty \frac{1}{r} \dif F(r),
\end{align*}
where the last step used the definition of $\THOLD$, implying that the
measure (under $\CDF$) of \Umr $r < \THOLD$ for which $\mathcal{P}$
incentivizes agents is the same as the measure of $r > \THOLD$ for
which $\mathcal{P}$ does not incentivize agents.

The final expression is the expected Lagrangian cost of
$\mathcal{P}'$, so we have shown that $\mathcal{P}'$ has no larger
cost than $\mathcal{P}$ in this one round.
\end{proof}

\begin{lemma} \label{lem:threshold}
The \lagobj of any time-expanded policy
$\mathcal{P}$ of $\ALG$ is (weakly) dominated by that of a threshold
policy.
\end{lemma}

\begin{proof}
Consider a time-expanded policy $\mathcal{P}$. 
Thus, the probability
$q$ with which $\mathcal{P}$ lets the agent play myopically
is the same in each round.
As in the proof of Lemma~\ref{lem:one-round-threshold}, we let
$\THOLD = \sup\{r| F(r) = q\}$.
Because $F(\THOLD) = q$, in each round, the reward of
$\mathcal{P}$ is the same as that of the 
threshold policy with threshold $\THOLD$.

Lemma~\ref{lem:one-round-threshold} establishes that in each round,
the cost of $\mathcal{P}'$ is no more than that of $\mathcal{P}$;
thus, the \lagobj for $\mathcal{P}'$ is at least as large as for $\mathcal{P}$.
\end{proof}


Because of Lemma \ref{lem:threshold}, it suffices to study the
optimal threshold policy, and determine the correct threshold.

As before, we let 
$x = \max_i \ExpectC{\reward{t}{i}}{\ve{S}_t,\Zed_{0:t-1}}$, and
$y = \ExpectC{\reward{t}{i_t^\ast}}{\ve{S}_t,\Zed_{0:t-1}}$.
The expected Lagrangian reward of $\TP{\THOLD}$ is 
\begin{align*}
x \cdot F(\THOLD) 
  + \int_{\THOLD}^\infty y-\frac{\lambda(x-y)}{r} \dif F(r)
& = y+ (x-y) \cdot (F(\THOLD)- \lambda \int_{\THOLD}^\infty \frac{\dif F(r)}{r}).
\end{align*}

Because $\CDF$ was assumed continuous, so is
$H(\THOLD) = F(\THOLD)- \lambda \int_{\THOLD}^\infty \frac{\dif F(r)}{r}$;
and because $H(0) < 0, \lim_{\THOLD \to \infty} H(\THOLD) = 1$,
the equation $F(\THOLD) = \lambda \int_{\THOLD}^\infty \frac{\dif F(r)}{r}$
has a solution $\THOLD$.
Fixing this choice of $\THOLD$, the expected Lagrangian payoff
simplifies to $y$.
The rest of the proof now proceeds as in the proof of Lemma~\ref{lem:TE},
yielding an approximation ratio of $1-F(\THOLD)\gamma$.
Writing $p^\ast = F(\THOLD)$ for the unconditional probability of a
myopic pull under $\TP{\THOLD}$, we obtain the same $(1-p^\ast\gamma)$
approximation ratio for the \lagobj as for the case of discrete signals.



\subsection{Upper Bound}
As for discrete signals, we next give a Diamonds-in-the-rough instance $\Dime$ on which the
upper bound for any policy matches the approximation ratio of the
threshold policy.
Consider the choice of the policy in the first round; it allows the
agent to play myopically with some probability $q$.
By Lemma~\ref{lem:one-round-threshold}, the optimal way to implement
this probability $q$ is to choose a threshold\footnote{Note that a
  priori, it is not clear that this threshold will not change in
  subsequent rounds; hence, we cannot yet state that a threshold policy is
  optimal.} $\THOLD$ and offer
incentives to the agent if and only if $r \geq \THOLD$.
Denote by $\THOLD^\ast$ the threshold we get for the time-expanded policy,
i.e., the solution of 
$F(\THOLD) = \lambda\int_{\THOLD}^\infty \frac{\dif F(r)}{r}$. 
Now, similar to Equation~\eqref{eq:diamond_opt}, the corresponding \lagobj is 
\begin{align}
  \label{eq:diamond_full}
    V =  (1 - \gamma + \gamma V) F(\THOLD) 
        + \int_{\THOLD}^\infty \left(
       \frac{1}{M} \cdot \theB \cdot M + (1 - \frac{1}{M})
    \gamma V - \lambda \frac{(1-\gamma)(1-\theB)}{r} \right) \dif F(r).
\end{align}

Letting $M \to \infty$ and solving Equation~\eqref{eq:diamond_full},
we obtain
\begin{align}
(1-\gamma)V
& = (1-\gamma - \theB) F(\THOLD) + \theB -\lambda
    (1-\gamma)(1-\theB) \int_{\THOLD}^{\infty}\frac{\dif F(r)}{r}.
\label{eq:fullinfo_opt}
\end{align}

Taking a derivative of Equation~\eqref{eq:fullinfo_opt} with respect
to $\THOLD$ suggests (note that the function may not be
differentiable, so this is merely used as a tool to suggest a useful
choice) that if we set $B$ to solve
$-\lambda(1-\gamma)(1 - B)=(1-\gamma - B)\theta^*$,
then our previously chosen threshold $\theta^*$ will be optimal for
the instance $\Dime$.



We next verify that this is indeed the case.
Substituting the value of $B$ into Equation~\eqref{eq:fullinfo_opt},
we obtain that maximizing Equation~\eqref{eq:fullinfo_opt} is
equivalent to minimizing  
$G(\THOLD) = \frac{1}{\THOLD^*} F(\THOLD) + \int_{\THOLD}^\infty
\frac{\dif F(r)}{r}$. 
Now, for any $\theta$, we have 

\begin{align*}
G(\THOLD^\ast) - G(\THOLD) 
&= \left(\frac{F(\THOLD^\ast)}{\THOLD^\ast} + \int_{\THOLD^\ast}^\infty
\frac{\dif F(r)}{r}\right) 
-  \left(\frac{F(\THOLD)}{\THOLD^\ast} + \int_{\THOLD}^\infty
\frac{\dif F(r)}{r} \right)\\
&= \frac{F(\THOLD^\ast) - F(\THOLD)}{\THOLD^\ast} +
\int_{\THOLD^\ast}^{\THOLD}\frac{\dif F(r)}{r} \\
&\le \frac{F(\THOLD^\ast) - F(\THOLD)}{\THOLD^\ast} +
\frac{F(\THOLD) - F(\THOLD^\ast)}{\THOLD^\ast} = 0.
\end{align*}

So $\THOLD^\ast$ is the maximizer of Equation~\eqref{eq:fullinfo_opt}
for the specific choice of $B$.


Thus, on this particular instance, the ratio achieved by our
threshold policy matches that of best possible policy. 

%% file: budget.tex
\section{Budgeted version}
\label{sec:budget}
In this section, we show matching lower and upper bounds for
maximizing the \Tet reward with a budget constraint.
Let $b$ be the fraction of $OPT_\gamma$ that the principal is allowed to use
(i.e., the principal's budget is $b \cdot OPT_\gamma$).
Recall that $p^\ast(\lambda)$ is the optimal value of the convex program \eqref{program:main}.


\begin{theorem}
Given budget $b \cdot OPT_\gamma$, there exists a policy whose approximation
ratio (with respect to $OPT_\gamma$) is
\begin{equation*}
\min_{\lambda} \{ 1 - p^\ast(\lambda) \gamma + \lambda b \}.
\end{equation*}
\end{theorem}
\begin{proof}
First, one can prove that for every $p \in [0, 1)$, 
there exists a $\lambda$ such that $p^\ast(\lambda) = p$.
Then, as every unconditional myopic probability $p$ can be achieved
using some $\lambda$,
we can follow the same approach in \citet{original} by taking the
limit of a sequence of probabilities $(p_n)_n$ that approaches
$p$. 
Here, the policy corresponding to each $p_n$ respects (or exhausts) the
budget, while the policy for $p$ exhausts (or respects) the budget. 
By suitable randomization between these policies, one can show that
the approximation ratio (with respect to $\opt_\gamma$) approaches
$\min_\lambda \{ 1 - p^\ast(\lambda) + \lambda  b\}$ in the limit.   
\end{proof} 
  
\begin{theorem}
Given budget $b \cdot OPT_\gamma$, the factor 
$\min_{\lambda} \{ 1 - p^\ast(\lambda) \gamma + \lambda b \}$ is
tight.
\end{theorem}
\begin{proof}
Again consider the class of Diamonds-in-the-rough instances.
Assume that the optimal policy that respects the budget is $\ALG$.
Thus $\DC(\ALG) \le b \cdot \opt_\gamma$ for every MAB instance.

Moreover, by Lemma \ref{lemma:imposs}, we know that for every
Lagrangian multiplier $\lambda$, the optimal 
policy for that $\lambda$ and its corresponding \Dime has \lagobj
value exactly $(1 - p^\ast(\lambda) \gamma) \cdot \opt_\gamma$. 

This means that for every Lagrangian multiplier $\lambda$ (and its
corresponding Diamonds-in-the-rough instance $\Dime$), we have 
the following equivalent inequalities for all $ \lambda \in (0,1)$,
\begin{align*}
\DR(\ALG) - \lambda \DC(\ALG) &\le (1 - p^\ast(\lambda)\gamma) \cdot \opt_\gamma\\
\DR(\ALG) &\le  (1 - p^\ast(\lambda) \gamma) \cdot \opt_\gamma +
            \lambda \cdot \DC(\ALG)\\
\DR(\ALG) &\le (1 - p^\ast(\lambda) \gamma +
\lambda b) \cdot \opt_\gamma
\end{align*}

So $$\DR(\ALG) \le \min_\lambda \{ (1 - p^\ast(\lambda) \gamma
+\lambda b) \} \cdot \opt_\gamma .$$
This matches the theorem above.
\end{proof}

%% file: concave-linear.tex
\section{Reducing Concave Money Value to Linear Money Value}
\label{sec:concave-linear}
\newcommand{\Done}[1][]{\ensuremath{\ifthenelse{\equal{#1}{}}{\mathcal{D}^1}{\mathcal{D}^1_{#1}}}\xspace}
\newcommand{\Dtwo}[1][]{\ensuremath{\ifthenelse{\equal{#1}{}}{\mathcal{D}^2}{\mathcal{D}^2_{#1}}}\xspace}
\def\scaleC{\ensuremath{C}\xspace}
\def\mDiff{\xi}

In this section, we show that among concave \MUFs, linear functions
obtain the worst-case approximation ratio. 
This shows that our assumption that agents have linear \MUFs is
without loss of generality in regard to worst-case approximation ratios. 
We fix the distribution (or prior) over agents and the signaling
scheme and then perform a worst-case analysis over MAB instances.
The intuition is that the rewards of an MAB instance can be scaled up
so large that only the asymptotic behavior of concave \MUFs matters.

For any \MUF $\MF$, define 
$r_{\MF} = \lim_{x \to \infty} \frac{\MF(x)}{x}$
to be the limiting slope\footnote{%
As $\MF$ is concave, this limit always exists and is finite.}
of $\MF$, 
and $d_{\MF} = \inf \Set{t \ge 0}{r_{\MF} x + t \geq \MF(x) \text{ for
    all } x \geq 0}$ the minimum intercept.
Thus, we ensure that for all $x$:
\[ 
r_{\MF} x \; \leq \; \Mf{x} \; \leq \; d_{\MF} + r_{\MF} x.
\]

Fix a signaling scheme $\sigSch$.
For a given signal $s$ and its original posterior distribution $F_s$
over concave functions, define two distributions over affine functions.
Under \Done[s], each function $\MF$ is mapped to the linear function
$x \mapsto r_{\MF} \cdot x$.
Under \Dtwo[s], the function $\MF$ is mapped to the function
$x \mapsto d_{\MF} + r_{\MF} \cdot x$.
We assume that \Done[s] is a continuous distribution for each $s$.

Our proof breaks into two parts: lower bound (algorithm) and upper
bound (impossibility). 

\textbf{Lower bound}: 
Given the true distribution $(F_s)_{s \in \sigSpace}$ over concave
functions, instead compute the optimal policy for 
$(\Done[s])_{s \in \sigSpace}$, which we denote by $\opt_1$.
Notice that $\opt_1$ is a time-expansion policy and only offers payment on the Gittins arm in each round.
 
Because $r_{\MF} \cdot x \leq \Mf{x}$ for all $x$, each agent likes money
more under $\MF$ than under $r_{\MF}$. 
Thus, we can couple the choices of agents between the two distributions:
each agent who plays non-myopically under the policy for $\Done$ can also
be incentivized to do so under $F$ using lower payments in expectation.
Therefore, the optimal policy for the true distribution can only do better.

Formally, let arm $i^\ast$ be the Gittins arm at some round $t$ during the execution of $\opt_1$.
Suppose that $\opt_1$ offers a payment of $p$ on arm $i^\ast$ and the arm will be pulled with probability $q$ when agents are drawn from $\Done[s]$.
When agents are drawn from $F_s$ instead, the probability $q_F$ that the arm will be pulled is no less than $q$ as agents like money more under $\MF$. Then our new policy for $F_s$ can randomize between offering the same payment $p$ or offering no payment so that the overall probability of pulling $i^\ast$ under $F_s$ is the same as that under $\Done[s]$.
In this way, we can couple the execution of $\opt_1$ and our new policy so that agents always make a same choices yet $\opt_1$ offers more payment in expectation.

\textbf{Upper bound}:
Next we show that in the worst case, the distribution $(F_s)$ over
concave functions does not yield a better approximation guarantee than
$(\Done[s])$.
First, notice that by the same argument as the previous paragraph, the
principal's utility under $(F_s)$ is upper bounded by her utility with
the distribution $(\Dtwo[s])_{s \in \sigSpace}$ over affine functions.
We will upper-bound the worst-case utility under $(\Dtwo[s])$ in terms
of that under $(\Done[s])$.

Thereto, we modify the Diamonds-in-the-rough instance from
Section~\ref{sec:diamond}, by scaling all reward values by some large
constant $\scaleC$. 
Let $\mDiff = C(1 - \gamma)(1 - \theB)$ be the difference between the
reward of the myopic arm and the Gittins arm. 
We will prove that for some carefully chosen constant $\scaleC$, the
optimal policy with payments for $\Dtwo$ cannot achieve much more than
the optimum for $\Done$. 
Let $(c_s)_s$ be the payments offered under different signals under
the optimal solution for $\Dtwo$.

By the argument preceding the program \eqref{program:dp},
the optimum solution for $\Dtwo$ maximizes
\begin{equation}
C (1 - \gamma) \cdot 
\sum_s \SP{s} \cdot \Prob[(r,d) \sim {\Dtwo[s]}]{r \cdot c_s + d \leq \mDiff}
+ \sum_s \SP{s} \cdot (1 - \Prob[(r,d) \sim {\Dtwo[s]}]{r \cdot c_s + d \leq \mDiff})(C \cdot B - \lambda c_s),
\label{eqn:dtwo-optimization}
\end{equation}
while the optimum solution for $\Done$ maximizes
\begin{equation}
C (1 - \gamma) \cdot 
\sum_s \SP{s} \cdot \Prob[r \sim {\Done[s]}]{r \cdot c_s \leq \mDiff}
+ \sum_s \SP{s} \cdot (1 - \Prob[r \sim {\Done[s]}]{r \cdot c_s \leq \mDiff})(C \cdot B - \lambda c_s).
\label{eqn:done-optimization}
\end{equation}

Now, for any arbitrarily small $\epsilon$, focus on a finite set
$\sigSpace'$ of signals which together contribute at least a $(1-\epsilon)$
fraction of the optimum utility under both $\Done$ and $\Dtwo$. 
Note that such a set must always exist, because we assumed that the
signal space is countable; by sorting the signals by their
contribution and taking prefixes in this ordering, we see that the
total utility of these prefixes must converge to the overall total
utility as the size of the prefix grows. 

For each individual signal $s$, the optimal $c_s \to \infty$
as $\scaleC \to \infty$. 
Thus, because $\sigSpace'$ is finite,
for any desired lower bound $\alpha$, there exists a $\scaleC$ 
such that $c_s \geq \alpha$ for all $s \in \sigSpace'$.
Similarly, for any $\epsilon > 0$, there is a (sufficiently large)
$\beta$ such that
$\Prob[(r,d) \sim {\Dtwo[s]}]{d \leq \beta} \geq 1 - \epsilon$ 
holds simultaneously for all $s \in \sigSpace'$.
Given a target $\epsilon$, we first choose a suitable $\beta$, and
then choose $\scaleC$ to ensure that $\alpha \gg \beta$.
Then, we get that for all $s \in \sigSpace'$,

\begin{align*}
\Prob[(r,d) \sim {\Dtwo[s]}]{r c_s + d \leq \mDiff} 
& = \Prob[(r,d) \sim {\Dtwo[s]}]{r \leq \frac{\mDiff - d}{c_s}}\\
& \geq (1-\epsilon) \cdot 
\Prob[(r,d) \sim {\Dtwo[s]}]{r \leq \frac{\mDiff}{c_s} - \frac{\beta}{\alpha}}\\
& = (1-\epsilon) \cdot 
\Prob[r \sim {\Done[s]}]{r \leq \frac{\mDiff}{c_s} - \frac{\beta}{\alpha}}.
\end{align*}

Because we assumed that $\Done[s]$ is continuous for each $s$, 
we obtain that for each $s$,
by choosing $\alpha \gg \beta$ large enough (ensured by making
$\scaleC$ large enough),
$\Prob[r \sim {\Done[s]}]{r \leq \frac{\mDiff}{c_s} - \frac{\beta}{\alpha}}
\geq \Prob[r \sim {\Done[s]}]{r \leq \frac{\mDiff}{c_s}} - \epsilon$.
By choosing $\scaleC$ as the maximum of the corresponding values, this
inequality holds simultaneously for all $s$.
In summary, we obtain that
\[
\Prob[(r,d) \sim {\Dtwo[s]}]{r c_s + d \leq \mDiff}
\geq 
\Prob[r \sim {\Dtwo[s]}]{r c_s \leq \mDiff} + 2\epsilon.
\]

Substituting this inequality into the objective value
\eqref{eqn:dtwo-optimization} shows that the objective values 
\eqref{eqn:dtwo-optimization} and \eqref{eqn:done-optimization}
differ by at most
\[
  \sum_s \SP{s} \cdot (|C(1 - \gamma - B)| \cdot 2\epsilon + \lambda \beta)
\; = \; |C(1 - \gamma - B)| \cdot 2\epsilon + \lambda \beta.
\]



The optimum value of the objective \eqref{eqn:done-optimization}
grows at least linearly in $C$. 
The reason is that when $C$ is scaled up by any constant $\nu$, a
feasible solution is obtained by scaling all $c_s$ up by $\nu$ as
well; this results in a multiplicative increase of $\nu$ in the
objective value. Thus, the \emph{best} $c_s$ values must attain at
least such an increase.
Because the error term $C(1 - \gamma - B) 2\epsilon + \lambda \beta$
is at most $O(\epsilon) \cdot \opt(\Done)$, we obtain that
\[
\opt(\Dtwo) \; \leq \; \opt(\Done) \cdot (1 + O(\epsilon)).
\]
Finally, adding in the $O(\epsilon)$ terms for signals not considered
in this argument does not change the above conclusion.
Making $\epsilon$ arbitrarily small (and scaling the
Diamonds-in-the-Rough instance correspondingly) then shows in the
limit that the distributions over linear and affine functions have the
same worst-case behavior.

%% file: garbling.tex
\section{Garbling of Signaling Schemes}
\label{sec:garbling}
In this section, we show that if we garble a signaling scheme
$\sigSch$ into $\sigSch'$, the optimal approximation ratio for the
garbled signaling scheme $\sigSch'$ cannot improve. 
By ``garbling,'' we mean (stochastically) mapping each signal in
$\sigSch$ to a random (new) signal in $\sigSch'$. 
\citet{marschak1968economic} gave the following\footnote{There are
  several equivalent definitions, of which we have chosen the one most
  suitable for our purposes.} formal definition (rephrased to fit our model):

\begin{definition}
Let $\sigSch, \sigSch'$ be signaling schemes with respective signal
spaces $\sigSpace, \sigSpace'$.
Then, $\sigSch'$ is a \emph{garbling} of $\sigSch$ if for all \Umrs $r$
and signals $s \in \sigSpace, s'\in \sigSpace'$:
$f_{s,s'}(r) = f_s(r)$, where $f_s(r)$ is the pdf of \Umr $r$
conditioned on signal $s$.
\end{definition}

As one can see, the garbled signaling scheme $\sigSch'$ contains less information than the original signaling scheme $\sigSch$. 
Recall that $1 - p^*(\psi) \gamma$ gives the optimal approximation ratio when the underlying (exogenous) signaling scheme is $\psi$. 
Intuitively, more informative signaling schemes should give raise to better approximation ratios, i.e. $1 - p^*(\sigSch) \gamma \ge 1 - p^*(\sigSch') \gamma$ if $\sigSch'$ is garbled from $\sigSch$.
Theorem~\ref{theorem:garbling_at_intro}, restated here, confirms this intution.




\begin{rtheorem}{Theorem}{\ref{theorem:garbling_at_intro}}
Let $\sigSch$ and $\sigSch'$ be two signaling schemes such that
$\sigSch'$ is a garbling of $\sigSch$.
Then, $1 - p^\ast(\sigSch) \gamma \geq 1 - p^\ast(\sigSch') \gamma$. 
\end{rtheorem}

\begin{proof}
Let $\SP{s}, \SP{s'}$ be the unconditional probabilities of observing
signals $s,s'$, as defined in Section~\ref{sec:prelim-signaling},
and let $(\MP{s'})_{s'\in \sigSpace'}$ be the optimal solution for
program \eqref{program:main} with signaling scheme $\sigSch'$.
We will show that the program \eqref{program:main} for signaling
scheme $\sigSch$ has a feasible solution $(\MP{s})_{s\in \sigSpace}$ 
such that 
$\sum_{s\in \sigSpace} \SP{s}\MP{s} =
\sum_{s'\in \sigSpace'} \SP{s'}\MP{s'}$.
This implies that the policy with the garbled signaling scheme cannot
outperform the policy with the original signaling scheme.

Let $\SP{s,s'}$ be the (joint) probability that signal $s$ is revealed
by $\sigSch$ and signal $s'$ is revealed by the garbling $\sigSch'$. 
Thus, $\SP{s'} = \sum_{s\in\sigSpace}\SP{s,s'}$.
Conditioned on the two revealed signals $s,s'$, let 
$\MP{s,s'}$ be the probability that the policy with parameters
$(\MP{s'})_{s' \in \sigSpace}$ (which only observes $s'$) 
obtains a myopic arm pull. Notice that this probability depends
on $s$: while $s$ does not affect the threshold that is set by the
policy, it does affect the distribution of the agent's \Umr, and hence
the probability of a myopic pull.
Thus, the overall probability of a myopic pull with signal $s'$ can be
obtained by summing over all (unobserved) signals $s$ as
$\SP{s'} \MP{s'} = \sum_{s\in \sigSpace}\SP{s,s'} \MP{s,s'}$.

Since $\MP{s'}$ is a feasible solution of \eqref{program:main}, we have 
\[
  \sum_{s'\in \sigSpace'} \SP{s'} \MP{s'} 
\geq \lambda\sum_{s'\in \sigSpace'} \SP{s'} \frac{1-\MP{s'}} {F_{s'}^{-1}(\MP{s'})}.
\]

Substituting $\SP{s'} = \sum_{s\in\sigSpace}\SP{s,s'}$
and $\MP{s'} = \frac{\sum_{s\in \sigSpace}\SP{s,s'} \MP{s,s'}}{\SP{s'}}$,
we can write

\[
\SP{s'} \cdot \frac{1-\MP{s'}}{F_{s'}^{-1}(\MP{s'})}
\; = \; \frac{\SP{s'}}{F_{s'}^{-1}(\MP{s'})} \cdot
  (1-\frac{\sum_{s\in \sigSpace}\SP{s,s'}\MP{s,s'}}{\SP{s'}})\\
\; = \; \sum_{s\in \sigSpace}
   \SP{s,s'} \cdot \frac{1-\MP{s,s'}}{F_{s'}^{-1}(\MP{s'})}.
\]

Next, we observe that $F_{s}^{-1}(\MP{s,s'}) = F_{s'}^{-1}(\MP{s'})$,
as follows:
because $\sigSch'$ is a garbling of $\sigSch$, we get that
the conditional distributions satisfy $F_s(r) = F_{s,s'}(r)$.
$F_{s'}^{-1}(\MP{s'}) =: \tau$ is a threshold chosen by the mechanism
with knowledge solely of $s'$, chosen to achieve a probability of
myopic play of exactly $\MP{s'}$.
$\MP{s,s'}$ denotes the probability of myopic play with this threshold
$\tau$, with the additional knowledge that the ungarbled signal was
$s$. 
Thus, we obtain that $\MP{s,s'} = F_{s,s'}(\tau) = F_s(\tau)$,
by the garbling property. 
Taking inverses now shows that
$F_{s'}^{-1}(\MP{s'}) = \tau = F_{s}^{-1}(\MP{s,s'})$.

Now define the target probabilities for the ungarbled signals as follows:
$\MP{s} =\sum_{s' \in \sigSpace'} \frac{\SP{s,s'}}{\SP{s}} \cdot \MP{s,s'}$.
We can then write

\begin{align*}
\sum_{s\in \sigSpace} \SP{s}\MP{s} 
& = \sum_{s \in \sigSpace} \sum_{s' \in \sigSpace'} \SP{s,s'}\MP{s,s'} \\
& = \sum_{s' \in \sigSpace'} \sum_{s \in \sigSpace} \SP{s,s'}\MP{s,s'} \\
& \geq \lambda \sum_{s' \in \sigSpace'} \sum_{s \in \sigSpace}
  \SP{s,s'} \cdot \frac{1-\MP{s,s'}}{F_{s}^{-1}(\MP{s,s'})}\\
& = \lambda \sum_{s \in \sigSpace} \sum_{s' \in \sigSpace'}
  \SP{s} \cdot \frac{\SP{s,s'}}{\SP{s}} \cdot
  \frac{1-\MP{s,s'}}{F_{s}^{-1}(\MP{s,s'})}\\
& \stackrel{(*)}{\geq} \lambda \sum_{s \in \sigSpace} \SP{s} \cdot 
  \frac{1-\sum_{s' \in \sigSpace'} \frac{\SP{s,s'}}{\SP{s}} \MP{s,s'}}{%
  F_{s}^{-1}(\sum_{s' \in \sigSpace'} \frac{\SP{s,s'}}{\SP{s}} \cdot \MP{s,s'})}\\
& = \lambda \sum_{s \in \sigSpace} 
  \SP{s} \cdot \frac{1-\MP{s}}{F_{s}^{-1}(\MP{s})}.
\end{align*}
Here, the inequality labeled (*) followed by semi-regularity of $F_s$.
The inequality derived above implies that $\MP{s}$ is a feasible
solution of \eqref{program:main} with signaling scheme $\sigSch$, and
it attains at least the same utility for the principal.
This completes the proof of the theorem.
\end{proof}

%% file: conclusions.tex
\section{Conclusions} \label{sec:conclusions}

We showed that the framework recently proposed by \citet{original} can
be generalized to the case when different agents have different
and non-linear tradeoffs for money vs.~utility derived from arm pulls.
While the generalized framework does not result in as clean a
characterization of feasible regions as the original work of
\citet{original}, it nonetheless holds true that time-expanded
versions of Gittins index policies are optimal in the worst case, and
that worst-case examples are of the simple ``Diamond-in-the-rough''
form.

We needed to assume a technical condition called semi-regularity for
our results. Whether time-expanded policies are optimal in the absence
of this condition is open.

There are many natural ways in which the model of \citet{original}
coul dbe further generalized. Perhaps most intriguingly, in the model,
agents only interact with the mechanism once, whereas it would be
natural to assume that the same agents return multiple times.
A natural model here would be one of the principal and just \emph{one}
agent, who has a different (steeper) time discount $\gamma' < \gamma$
than the principal, and must be incentivized to pull arms with more
foresight. In this sense, we analyzed the special case $\gamma' = 0$.

This direction appears quite a bit more difficult to analyze.
If agents may return more than once, this opens the door for strategic
behavior; it seems possible that an agent may choose a particular arm
to pull to help or prevent the principal from learning, in turn
affecting possible future payments. This makes this model quite a bit
more complicated to analyze.

If one were to try and take the model into a direction of more
realism, one could consider specifics of some of the applications
listed in the introduction, such as that of an online retailer. 
In that case, one may have to account for the fact that agents who
receive a payment (i.e., discount) on a product may alter their
perception or rating of the product.

%% file: app-missing-proofs.tex
\section{Missing Proofs}
\label{sec:missing-proofs}
\begin{rtheorem}{Lemma}{\ref{lemma:semireg}}
Let $G$ be the CDF of a non-negative random variable,
and define $\inv{G}(x) := sup \{ t \ge 0 : G(t) \le x \}$. 
If $\inv{G}(x) (1 - x)$ is concave, then
$\frac{1 - x}{\inv{G}(x)}$ is convex.
In particular, regularity implies semi-regularity.
\end{rtheorem}

\begin{proof}
  We only need to look at 3 points $x$, $\alpha x + (1 - \alpha)y$ and $y$, where $0 \le x < y \le 1$.
  Let $z = \alpha x + (1 - \alpha) y$
  By concavity of $\inv{G}(x) (1 - x)$, we have
  \begin{equation}
    \inv{G}(z) (1 - z) \ge \alpha \inv{G}(x)(1 - x) + (1 - \alpha)\inv{G}(y) (1 - y)
  \end{equation}
  Rearranging, we have:
  \begin{equation}
    0 \ge \alpha (1 - x) \frac{\inv{G}(x) - \inv{G}(z)}{\inv{G}(z)}
    + (1 - \alpha)(1 - y) \frac{\inv{G}(y) - \inv{G}(z)}{\inv{G}(z)}
  \end{equation}
  By monotonicity of $\inv{G}$, $\inv{G}(x) \le \inv{G}(z) \le \inv{G}(y)$.
  This implies
  \begin{equation}
    \label{app:semi}
    0 \ge \alpha (1 - x) \frac{\inv{G}(x) - \inv{G}(z)}{\inv{G}(x)}
    + (1 - \alpha)(1 - y) \frac{\inv{G}(y) - \inv{G}(z)}{\inv{G}(y)},
   \end{equation}
   as $\inv{G}(x) \le \inv{G}(z)$ and $\inv{G}(y) \ge \inv{G}(z)$.
   One can see that inequality \eqref{app:semi} is equivalent to the convexity of $\frac{1 - x}{\inv{G}(x)}$.
\end{proof}

The proofs are very straightforward (syntactic) modifications of those of the corresponding lemmas in \cite{original}, we include them here for completeness.

\begin{rtheorem}[Modification of Lemma 4.2 of \cite{original}]{Lemma}{\ref{cite42}}
Given a parameter $\lambda$ and a signaling scheme $\sigSch$. Let $\zeta_{t-1}=\sum_{t'<t}Z_{t'}$ be the total
number of non-myopic steps performed by the time-expanded algorithm $\Te{\MPV}{A}$
prior to time $t$, where $\MPV$  satisfies
$\sum_{s \in \sigSpace} \SP{s} \MP{s} \ge\lambda\sum_{s \in \sigSpace}
\SP{s} \frac{1 -\MP{s}}{F_s^{-1} (\MP{s})}$. 
Then, for any $0\le n\le t$,
\begin{align*}
\ExpectC[\Te{\MPV}{A}]{\reward{t}{i_t}-\lambda\payment{t}{i_t}}{\zeta_{t-1} = n}
& \ge \Expect[\ALG]{\reward{n}{i_n}}.
\end{align*}
\end{rtheorem}

\begin{proof}
The Lagrangian utility at time $t$ is 
$\reward{t}{i_t} - \lambda\payment{t}{i_t}$.
Let $x = \max_i \ExpectC{\reward{t}{i}}{\ve{S}_t,\Zed_{0:t-1}}$ 
and $y = \ExpectC{\reward{t}{i_t^\ast}}{\ve{S}_t,\Zed_{0:t-1}}$, 
where $\Zed_{0:t-1}=(\Zed_0,\Zed_t,\dots,\Zed_{t-1})$. 
Since the myopic arm is played with probability $\MP{s}$,
the expected Lagrangian utility will be 
\begin{align*}
\ExpectC[{\Te{q}{A}}]{\reward{t}{i_t} - \lambda \payment{t}{i_t}}{\ve{S}_t, \Zed_{0:t-1}}&= \sum_{s\in\sigSpace}\SP{s}(\MP{s} x +
(1-\MP{s})(y-\lambda\frac{x-y}{F_s^{-1}(\MP{s})}) )\\
&= y\sum_{s\in \sigSpace} \SP{s} + (x-y)\sum_{s\in\sigSpace}\SP{s}(\MP{s} -
\frac{\lambda(1-\MP{s})}{F_s^{-1} (\MP{s})}) \ge y.
\end{align*}

The last inequality above is due to $x\ge y$ 
(by the myopic arm preference) and 
$\sum_{s \in \sigSpace} \SP{s} \MP{s} \ge\lambda\sum_{s \in \sigSpace}
\SP{s} \frac{1 - \MP{s}}{F_s^{-1} (\MP{s})}$ (by assumption of the lemma).
Notice the condition 
$\sum_{s \in \sigSpace} \SP{s} \MP{s} \ge\lambda\sum_{s \in \sigSpace} \SP{s} \frac{1 - \MP{s}}{F_s^{-1} (\MP{s})}$ in the
statement of the lemma allowed us the cancel out the myopic
arm-rewards $x$, which would otherwise have been difficult to analyze.
Now we have $$\ExpectC[\Te{\MPV}{A}]{\reward{t}{i_t}- \lambda \payment{t}{i_t}}{\ve{S}_t, \Zed_{0:t-1}} \geq
\ExpectC[\Te{\MPV}{A}]{\reward{t}{i_t^\ast}}{\ve{S}_t, \Zed_{0:t-1}}.$$

Since the expected reward only depends on the current state of arm
$\ve{S}_t$, we can add the status of previous round to the condition, 
$$\ExpectC[\Te{\MPV}{A}]{\reward{t}{i_t}- \lambda \payment{t}{i_t}}{\ve{S}_{0:t}, \Zed_{0:t-1}} \geq
\ExpectC[\Te{\MPV}{A}]{\reward{t}{i_t^\ast}}{\ve{S}_{0:t},
  \Zed_{0:t-1}}.$$

Taking conditional expectation on both sides with the condition
$\hat{\ve{S}}_t, \zeta_{t-1} = n$,

\begin{align*}
 \ExpectC[\Te{\MPV}{A}]{\ExpectC[\Te{\MPV}{A}]{\reward{t}{i_t}- \lambda
   \payment{t}{i_t}}{\ve{S}_{0:t}, \Zed_{0:t-1}}}{\hat{\ve{S}}_t,\zeta_{t-1} = n} \\
\geq
\ExpectC[\Te{\MPV}{A}]{\ExpectC[\Te{\MPV}{A}]{\reward{t}{i_t^\ast}}{\ve{S}_{0:t}, \Zed_{0:t-1}}}{\hat{\ve{S}}_t, \zeta_{t-1} = n}.
\end{align*}

Since $\hat{\ve{S}}_t, \zeta_{t-1} = n$ is measurable by
$\ve{S}_{0:t}, \Zed_{0:t-1}$, we can use the tower property of
conditional expectations and get
$$\ExpectC[\Te{\MPV}{A}]{\reward{t}{i_t}- \lambda
  \payment{t}{i_t}}{\hat{\ve{S}}_t,\zeta_{t-1} = n} \geq
\ExpectC[\Te{\MPV}{A}]{\reward{t}{i_t^\ast}}{\hat{\ve{S}}_t, \zeta_{t-1}
  = n}.$$

Let $\tau\ge t$ be the next time when $\Te{\MPV}{A}$ performs a
non-myopic pull. Since the reward sequence for each arm forms a
Martingale, and $i_u^\ast$ is fixed with a certain $\hat{\ve{S}}_t$,
we have that $\ExpectC[\Te{\MPV}{A}]{\reward{u}{i_{u}^\ast}}{\hat{\ve{S}}_t,
\zeta_{t-1} = n}$ is identical for all $u\ge t$. Writing $\ExpectC[\Te{\MPV}{A}]{\reward{\tau}{i_{\tau}^\ast}}{\hat{\ve{S}}_t,
\zeta_{t-1} = n}$ with the law of total probability, we have

\begin{align*}
 \ExpectC[\Te{\MPV}{A}]{\reward{\tau}{i_{\tau}^\ast}}{\hat{\ve{S}}_t,\zeta_{t-1}
  = n} &= \sum_{u\ge t}\Prob{\tau = u}\cdot
  \ExpectC[\Te{\MPV}{A}]{\reward{u}{i_{u}^\ast}}{\hat{\ve{S}}_t,\zeta_{t-1}
  = n, \tau = u} \\
&  = \ExpectC[\Te{\MPV}{A}]{\reward{t}{i_{t}^\ast}}{\hat{\ve{S}}_t,\zeta_{t-1}
  = n} .
\end{align*}

So far, we have proved that $$\ExpectC[\Te{\MPV}{A}]{\reward{t}{i_t}- \lambda
  \payment{t}{i_t}}{\hat{\ve{S}}_t,\zeta_{t-1} = n} \geq
\ExpectC[\Te{\MPV}{A}]{\reward{\tau}{i_\tau^\ast}}{\hat{\ve{S}}_t, \zeta_{t-1}
  = n}.$$

Taking conditional expectations with respect to $\zeta_{t-1} = n$, then
applying the law of iterated expectation, 
$$\ExpectC[\Te{\MPV}{A}]{\reward{t}{i_t}- \lambda
  \payment{t}{i_t}}{\zeta_{t-1} = n} \geq
\ExpectC[\Te{\MPV}{A}]{\reward{\tau}{i_\tau^\ast}}{\zeta_{t-1}
  = n}.$$
Note the right-hand side is exactly $\Expect[\ALG]{\reward{n}{i_n}}$ by
definition of $\Te{\MPV}{A}$,
finishing the proof.
\end{proof}

\begin{rtheorem}[Variation of Lemma 3.2 of \cite{original}]{Lemma}{\ref{cite32}}
Given a parameter $\lambda$ and a signaling scheme $\sigSch$.  Assume  $\MPV$ satisfies
$\sum_{s \in \sigSpace} \SP{s} \MP{s} \ge\lambda\sum_{s \in \sigSpace}
\SP{s} \frac{1 -\MP{s}}{F_s^{-1} (\MP{s})}$, then for $\eta =
\frac{(1-p)\gamma}{1-p\gamma}$, where $p=\sum_{s\in \sigSpace}\SP{s}\MP{s}$, we have 
\begin{align*}
R_{\lambda}^{(\gamma)}(\Te{\MPV}{A})
& \ge \frac{1-\eta}{1-\gamma}\cdot R^{(\eta)}(\ALG).
\end{align*}
\end{rtheorem}

\begin{proof}
  \begin{align*}
    R_{\lambda}^{(\gamma)}(\Te{\MPV}{A})\quad &=\quad &&\sum_{t=0}^\infty \gamma^t
                                        \Expect[\Te{\MPV}{A}]{\reward{t}{i_t}-
                                        \lambda\payment{t}{i_t}}\\
    &= &&\sum_{t=0}^\infty\sum_{n=0}^\infty\gamma^t \ExpectC[\Te{\MPV}{A}]{\reward{t}{i_t}- \lambda
  \payment{t}{i_t}}{\zeta_{t-1} =n} \cdot \Prob{\zeta_{t-1} = n}\\
&\stackrel{\mathclap{\mbox{Lemma~\ref{cite42}}}}{\ge}&&\sum_{n=0}^\infty
  \Expect[\ALG]{\reward{n}{i_n}}\cdot\sum_{t=0}^\infty\gamma^t\cdot
  \Prob{\zeta_{t-1}=n}\\
&=&&\sum_{n=0}^\infty  \Expect[\ALG]{\reward{n}{i_n}}\cdot \sum_{t=0}^\infty\gamma^t\cdot
  \binom{t}{n}\cdot (1-p)^np^{t-n}\\
&=&&\sum_{n=0}^\infty  \Expect[\ALG]{\reward{n}{i_n}}\cdot \gamma^n
  (1-p)^n \cdot \sum_{i=0}^\infty
  \binom{n+i}{n}\cdot (\gamma p)^i\\
&=&&\sum_{n=0}^\infty  \Expect[\ALG]{\reward{n}{i_n}}\cdot \gamma^n
  (1-p)^n \cdot (1-\gamma p)^{-(n+1)}\\
&=&&\sum_{n=0}^\infty  \Expect[\ALG]{\reward{n}{i_n}}\cdot
  \frac{1-\eta}{1-\gamma}\cdot \eta^n.
  \end{align*}

In the penultimate step, we use $\sum_{i=0}^\infty \binom{n+i}{n} x^i
= (1-x)^{-(n+1)}$ for $n \ge 0$ and $|x|< 1$. 
\end{proof}

%% file: rewarddep.tex
\section{Impossibility of Exploration with Reward-Dependent Payments}
\label{app:rewarddep}

In this section, we show that if the payment scheme is a function only of the arm reward obtained by the agent, it can be impossible to incentivize the agent to pull the optimal arm. 

We use a slightly modified ``Diamonds in the Rough'' instance \Dime,
by changing the myopic arm to two arms. Both of them produce
i.i.d. rewards from a known distribution. The distributions are showed
in Table~\ref{tab:dis}:
\begin{table}[h]
  \centering

\begin{tabular}[h]{c|c| c}
Reward &  Arm 1 outputs reward with probability & Arm 2 outputs reward with probability\\
\hline
$(1-\gamma)B\cdot M$ & $(1+\epsilon)/M$ & $(1-\epsilon)/M$\\
0 &  $1- (1+\epsilon)/M$&  $1-(1-\epsilon)/M$
\end{tabular}
    
  \caption{Distributions of myopic arms}
  \label{tab:dis}
\end{table}

We emphasize that this is not the type of arm with degenerate distribution that produces an initially unknown constant reward. 

Also we keep the infinite supply of arms with degenerate distribution. Recall the definition as follows,

\begin{enumerate}
 \item With probability $1/M$, the arm's reward is a degenerate distribution of the constant $(1 - \gamma) \theB \cdot M$
 (good state);
 \item With probability $1 - 1/M$, the arm's reward is a degenerate
   distribution of the constant $0$ (bad
 state).
\end{enumerate}

Because the payment function can depend only on the reward that the agent obtained, it is completely characterized by the payment $p_H$ in response to obtaining $(1-\gamma)B\cdot M$ and the payment $p_L$ in response to obtaining $0$.
When $(1-\gamma)B\cdot M + p_H\ge p_L$,   we have
$$((1-\gamma)B\cdot M + p_H)\cdot\frac{1+\epsilon}{M}+ p_L\cdot \left(1-\frac{1+\epsilon}{M}\right) \ge ((1-\gamma)B\cdot M + p_H)\cdot\frac{1}{M}+ p_L\cdot (1-\frac{1}{M});$$
otherwise, we have 
$$((1-\gamma)B\cdot M + p_H)\cdot\frac{1-\epsilon}{M}+ p_L\cdot \left(1-\frac{1-\epsilon}{M}\right) \ge ((1-\gamma)B\cdot M + p_H)\cdot\frac{1}{M}+ p_L\cdot (1-\frac{1}{M}).$$

So the expected utility of pulling the optimal arm is always weakly dominated by one of two myopic arms. Therefore, there is no way to incentivize the agent to pull the optimal arm. 
 